\documentclass[final,onefignum,onetabnum]{siamart171218}
\usepackage{amsfonts}
\usepackage{epstopdf}
\ifpdf
  \DeclareGraphicsExtensions{.eps,.pdf,.png,.jpg}
\else
  \DeclareGraphicsExtensions{.eps}
\fi

\newsiamremark{remark}{Remark}
\newsiamremark{hypothesis}{Hypothesis}
\crefname{hypothesis}{Hypothesis}{Hypotheses}
\newsiamthm{claim}{Claim}

\title{An efficient approach for nonconvex semidefinite optimization via customized alternating direction method of multipliers %\thanks{Submitted to the editors July 2020.}
}

\author{Chuangchuang Sun\thanks{Department of Aerospace Engineering, Mississippi State University(\email{csun@ae.msstate.edu}).}
}

\usepackage{amsopn}

\usepackage{amsmath}
\usepackage{amssymb}
\usepackage{graphicx,epstopdf}
\usepackage{epsfig}
\usepackage{color}
\usepackage{enumerate}
\usepackage{graphicx}
\usepackage{multirow}

\renewcommand{\theequation}{\arabic{section}.\arabic{equation}}

\newcommand{\bea}{\begin{eqnarray}}
\newcommand{\eea}{\end{eqnarray}}
\newcommand{\beas}{\begin{eqnarray*}}
\newcommand{\eeas}{\end{eqnarray*}}
\newcommand{\leftm}{\left[\begin{array}}
\newcommand{\rightm}{\end{array}\right]}
\newcommand{\reals}{\mathbb R}
\newcommand{\symm}{\mathbb S}

\newcommand{\zeros}{\mathbf{0}}

\newcommand{\mC}{\mathcal{C}}

\newcommand{\R}{\mathbb{R}}
\newcommand{\mL}{\mathcal{L}}
\newcommand{\mA}{\mathcal{A}}

\newcommand{\mN}{\mathcal{N}}
\newcommand{\proj}{\mathbf{proj}}

\newcommand{\minimize}[1]{\underset{#1}{\mathrm{min}}\ }
\newcommand{\maximize}[1]{\underset{#1}{\mathrm{max}}\ }
\newcommand{\st}{\textrm{\ \ s.t.\ \ }}

\newcommand{\argmin}[1]{\underset{#1}{\mathrm{argmin}}}

\newcommand{\tr}{\mathbf{Tr}}
\newcommand{\diag}{\mathbf{diag}}
\newcommand{\Diag}{\mathbf{Diag}}

\newcommand{\mb}{\mathbf}
\newcommand{\rank}{\mathbf{rank}}
\newcommand{\prox}{\mathbf{prox}}
\newcommand{\sign}{\mathbf{sign}}
\newcommand{\blkdiag}{\mathbf{blkdiag}}
\newcommand{\ceil}[1]{\left \lceil{#1}\ \right \rceil }

\newcommand{\eg}{\textit{e.g.}}
\newcommand{\ie}{\textit{i.e.}}

\newcommand{\bmat}{\left[\begin{matrix}}
\newcommand{\emat}{\end{matrix}\right]}

\usepackage{algpseudocode}% http://ctan.org/pkg/{algorithms,algorithmx}

\algnewcommand{\Inputs}{%
	\State \textbf{Inputs:}
}
\algnewcommand{\Initialize}{%
	\State \textbf{Initialize:}
}
\algnewcommand{\Outputs}{%
	\State \textbf{Outputs:}
}

\ifpdf
\hypersetup{
  pdftitle={ADMM for nonconvex semidefinite optimization},
  pdfauthor={C. Sun}
}
\fi

\begin{document}

\maketitle

% REQUIRED
\begin{abstract}
  We investigate a class of general combinatorial graph problems, including MAX-CUT and community detection, reformulated as quadratic objectives over nonconvex constraints and solved via the alternating direction method of multipliers (ADMM).
  We propose two reformulations: one using vector variables and a binary constraint, and the other further reformulating the Burer-Monteiro form for simpler subproblems.
  Despite the nonconvex constraint, we prove the ADMM iterates converge to a stationary point in both formulations, under mild assumptions.
  Additionally, recent work suggests that in this latter form, when the matrix factors are wide enough, local optimum with high probability is also the global optimum.
  To demonstrate the scalability of our algorithm, we include results for MAX-CUT, community detection, and image segmentation  benchmark and simulated examples.

\end{abstract}

% REQUIRED

% REQUIRED

\section{Introduction}

We consider rank-constrained semidefinite optimization problems (SDPs) of the type
\begin{equation}
%\begin{array}{ll}
\minimize{Z,X}\  f(Z),\  \st\   \mA(Z) = b,Z = XX^T,X\in \mC
%\end{array}
\label{eq:main}
\end{equation}
where the matrix variable
 $Z\in\symm_+^n$ is a $n\times n$ symmetric semidefinite matrix,  and $X\in \R^{n\times r}$ a low rank symmetric factor.
The linear constraints $\mA(Z) =b$  constrain either the diagonal or trace of $Z$, and the set $\mC$ controls desirable features of the factor--e.g. nonnegativity, integer, norm-1, etc. ($\mC$ may be nonconvex.)
The objective function $f(x)$ is convex, differentiable everywhere, with $L_f$-Lipschitz gradient, but the overall problem \eqref{eq:main} is nonconvex.

This problem is equivalent to many important nonconvex SDPs, such as  the MAX-CUT problem and its related applications \cite{bandeira2016low,fortunato2016community,javanmard2016phase},  rank-constrained nonnegative matrix factorization problem \cite{gillis2011nonnegative,ding2005equivalence}, and constrained eigenvalue problems \cite{da2010cone,gander1991constrained,judice2007eigenvalue}.
It is known that exactly solving \eqref{eq:main} globally is in general a very difficult problem, as it includes many NP-Hard problems.
Methods for heuristically solving \eqref{eq:main} fall in three categories: i)solving the convexified SDP, where \eqref{eq:main} does not have the rank-$r$ or $X\in \mC$ constraint, using any convex optimization method \cite{goemans1995improved, helmberg2000spectral, fujie1997semidefinite}, ii) approximatly solving \eqref{eq:main} using an alternating minimization method  \cite{burer2003nonlinear,boumal2016non} and relying on statistical arguments suggesting that the  acquired local optimal = the global optimal \cite{boumal2016non}, or iii) using other application-specific approaches \cite{lee2001algorithms, fortunato2016community}.
The methods investigated in this paper fall in the second category.
Specifically, we investigate solving \eqref{eq:main} using ADMM and linearized ADMM on two reformulations.
We find that these flexible reformulations allow easy incorporation of low-rank and sparse structure, making the resulting algorithm extremely scalable, in both memory and computation, which we demonstrate on a number of popular applications.

However, often nonconvex formulations of SDPs are not favored because the convergence behavior of standard algorithms are not well understoodn. Specifically, an iterative procedure can do one of four things: diverge, oscillate within a bounded interval, converge to an arbitrary point, or converge to a useful point.
 We show that  linearized ADMM on a nonsymmetric reformulation of \eqref{eq:main} can either converge to a stationary point, or diverge to $\pm \infty$; it cannot oscillate or converge to a non-stationary point. Additionally, for the case without linear constraints, vanilla ADMM  is guaranteed to converge to a stationary point with a monotonically decreasing augmented Lagrangian term, and at a linear rate if the objective is strongly convex.

\section{Applications}
It is well-known that many convex optimization problems can be reformulated as SDPs (e.g. \cite{wolkowicz2012handbook}).
In nonconvex optimization, SDPs are studied in several key areas, as tight convex relaxations of otherwise NP-hard problems.

\subsection{Combinatorial problems}
A simple reparametrization of the constraint $x\in \R^n$, $x_i\in \{-1,1\}$ is as $X = xx^T$, $\diag(X) = 1$. This property has been heavily exploited for finding lower bounds in combinatorial optimization \cite{laurent2009sums, rendl2012semidefinite, goemans1995improved}, and generalized further to polynomial optimization
\cite{blekherman2012semidefinite, anjos2012introduction}.
Of high interest is the MAX-CUT problem
%For an undirected graph with $n$ nodes and (possibly weighted) edges between some pairs of nodes, a \textit{cut} is defined as a partition of the $n$ nodes, and the \textit{cut value} is the sum of the weights of the severed edges. The MAX-CUT problem is to find the cut in a given graph that gives the maximum cut value.
\begin{equation}
%\begin{array}{ll}
\minimize{x\in \R^n}\  x^TCx,\ \st \ x_i\in \{-1,1\},\quad i = 1,\hdots, n
%\end{array}
\label{eq:comb}
\end{equation}
where $C = (A - \diag(A\mb 1))/4$ and $A\in \symm^n$ is the symmetric adjacency matrix of an undirected graph. Written in this way, the solution to \eqref{eq:main} is exactly the maximum cut of an undirected graph with nonnegative weights $A_{ij}$.

This seemingly simple framework appears in many other applications, such as community detection \cite{abbe2016exact} and image segmentation \cite{shi2000normalized}, and is equivalent to the nonconvex SDP
\begin{equation}
%\begin{array}{ll}
\minimize{Z}\  \tr(CZ),\ \st\  Z_{kk} = 1, Z \succeq 0, \rank(Z) = 1.
%\end{array}
\label{eq:comb-SDP}
\end{equation}
Lifting $x\in \R^n$ to a skinny matrix $X\in \R^{n\times k}$ generalizes this technique to  partitioning \cite{karisch1998semidefinite} and graph coloring problems
\cite{karger1998approximate}.

\paragraph{Related works on MAX-CUT}
More generally, combinatorial methods can be solved using branch-and-bound schemes, using a linear relaxation of \eqref{eq:main} as a bound  \cite{barahona1988application,de1995exact}, where the binary constraint $x\in \{-1,1\}$ is relaxed to $0 \leq (x+1)/2 \leq 1$. Historically, these ``polyhedral methods" were the main approach to find exact solutions of the MAX-CUT problem. Though this is an NP-Hard problem, if the graph is sparse enough, branch-and-bound converges quickly even for very large graphs \cite{de1995exact}. However, when the graph is not very sparse, the linear relaxation is loose, and finding efficient branching mechanisms is challenging, causing the algorithm to run slowly.
The MAX-CUT problem can also be approximated by one pass of the linear relaxation (with bound
$\frac{f_{\text{relax}}}{f_{\text{exact}}} \geq 2 \times \#$edges) \cite{poljak1994expected}.

A tighter approximation can be found with the semidefinite relaxation, which is also used for better bounding in branch-and-bound techniques \cite{helmberg1998solving, rendl2007branch, burer2008finite, bao2011semidefinite,krislock2012improved}.
In particular, the rounding algorithm of \cite{goemans1995improved} returns a feasible $\hat x$ given optimal $Z$, and is shown in expectation to satisfy $\frac{x^TCx}{\hat x^TC\hat x} \geq 0.878$. For this reason, the semidefinite relaxation for problems of type \eqref{eq:main} are heavily studied (\eg \cite{poljak1995recipe, helmberg2000semidefinite, fujie1997semidefinite}).

\paragraph{Specialization to community detection}
A small modification of the matrix $C$ generalizes problems of form \eqref{eq:comb} and \eqref{eq:comb-SDP} to community detection in machine learning.
Here the problem is to identify node clusters in undirected graphs that are more likely to be connected with each other than with nodes outside the cluster.
This prediction  is useful in many graphical settings, such as interpreting online communities through social network or linking behavior \cite{papadopoulos2012community}, interpreting biological ecosystems \cite{girvan2002community},  finding disease sources in epidemiology\cite{keeling2005implications}, and many more.
There are many varieties and methodologies in this field, and it would be impossible to list them all, though many comprehensive overviews exist (e.g. \cite{fortunato2016community}).

The stochastic binary model \cite{holland1983stochastic} is one of the simplest generative models for this application. Given a graph with $n$ nodes and parameters $0 < q < p < 1$, the model partitions the nodes into two communities, and generates an edge between nodes in a community with probability $p$ and nodes in two different communities with probability $q$. Following the analysis in \cite{abbe2016exact}, we can define  $C = \frac{p + q}{2}\mb1 \mb1^T - A$, where $A$ is the graph adjacency matrix, and the solution  to \eqref{eq:main} gives a solution to the community detection problem with sharp recovery guarantees.

\subsection{Nonnegative factorization}
For a symmetric matrix $C$, the maximum eigenvalue / eigenvector pair of  $C$ is the solution to the nonconvex optimization problem
\begin{equation}
%\begin{array}{ll}
\maximize{x\in \R^n} \  x^TCx, \ \st\   \|x\|_2 = 1.
%\end{array}
\label{eq:eigenvalue}
\end{equation}
By inverting the sign of $C$, we can transform this into a minimization problem, or equivalently acquire the minimum eigenvalue/eigenvector pair.
Interestingly, despite the nonconvex nature of \eqref{eq:eigenvalue}, we have many efficient globally optimal methods for finding $x$, e.g. Lanczos, Arnoldi, etc.  However, adding any additional constraints, such as nonnegativity of $x$ \cite{queiroz2004symmetric}, and simple methods generally do not work without heavy data assumptions \cite{deshpande2014cone}.
This is of interest in problems such as phase retrieval, recommender systems with positive-only observations, clustering and topic models, etc.
Here we discuss three variations of the nonnegative factorization problem appearing in literature, all of which are special instances of \eqref{eq:main}.

\paragraph{Optimization over spectrahedron}
We can frame \eqref{eq:eigenvalue} as a linear objective over the \emph{spectrahedron}
\begin{equation}
%\begin{array}{ll}
\minimize{Z\in \symm^n} \  \tr(CZ), \ \st\  \tr(Z) = 1, Z \succeq 0.
%\end{array}
\label{eq:eig-spectrahedron}
\end{equation}
If additionally the maximum eigenvalue of $C$ is isolated (corresponding only to one leading eigenvector) then $Z = xx^T$ and $Cx = \lambda_{\max}(C) x$.
To see this, by definition,
\[
\lambda_{\max}(C) = \max_{x : \|x\|_2 = 1} x^TCx = \max_{Z : Z = xx^T, \|x\|_2 = 1} \tr(CZ) = \max_{Z : \tr(Z) = 1,X\succeq 0} \tr(CZ).
\]
As a consequence, note that that though \eqref{eq:eig-spectrahedron} is convex, the solution $Z^*$ will always have rank 1 when $\lambda_{\max}(C)$ has multiplicity 1.
A simple extension of \eqref{eq:eig-spectrahedron} often used in nonnegative PCA \cite{zass2007nonnegative} is
\begin{equation}
\begin{array}{ll}
\minimize{Z\in \symm^n,x\in \R^n} \  \tr(CZ),\ \st\,  \tr(Z) = 1, Z \succeq 0,Z = xx^T,\quad x \geq 0,
\end{array}
\label{eq:nonneg-eigenvalue}
\end{equation}
which is an instance of \eqref{eq:main} with $\mC$ the nonnegative orthant.

\paragraph{Factorization with partial observations}
An equivalent way of formulating the top-$k$ nonnegative-eigenvector problem is as the nonnegative minimizer $X$ to $\|XX^T-C\|_2$ where $X$ is $\R^{n\times k}$. However, in many applications, we may not have full view of the matrix $C$, (e.g. $C$ is  a rating matrix). Suppose that an index set $\Omega$ defines the observed entries, e.g. $\{i,j\}\in \Omega$ implies $C_{ij}$ is known. Then the nonnegative factorization problem can be written as
\begin{equation}
%\begin{array}{ll}
\minimize{Z\in \symm^n,x\in \R^n} \ \displaystyle\sum_{i,j\in \Omega} (Z_{ij}-C_{ij})^2, \ \st \  Z = xx^T,\quad x\geq 0
%\end{array}
\label{eq:nonneg-SDP}
\end{equation}
This formulation exists in \cite{lee1999learning}.

\paragraph{Projective Nonnegative Matrix Factorization}
A third method toward this goal is to optimize over the low rank projection matrix itself~\cite{yuan2005projective}, a variant of nonnegative matrix factorization, solving
\begin{equation}
%\begin{array}{ll}
\min_{Z\in \symm^n,X\in \R^{n\times k}}\ \|B-ZB\|_2, \ \st\   Z  = XX^T, \  X\geq 0
%\end{array}
\label{eq:projMF}
\end{equation}
Here, the data matrix may not even be symmetric, but $\frac{1}{\tr(Z)}ZB$ will approximate the projection of $B$ to its top-$k$ singular vectors.
%\red{Is it obvious why we want $X\geq 0$ here? Also is this easy to do in implementation?}

\section{Related work}

\paragraph{Convex relaxations}
If $r = n$ and $\mC = \symm^n$ then \eqref{eq:main} is a convex problem, and can be solved using many conventional methods with strong convergence guarantees. However, even in this case, if $n$ is large, traditional semidefinite solvers are computationally limiting. In the most general case,  an interior point method solves at each iteration a KKT system of at least order $n^6$, and most first-order methods for general SDPs require eigenvalue decompositions, which are of order $O(n^3)$ per iteration.

\paragraph{Low-rank convex cases}
In fact, assuming low-rank solutions often allows for the construction of faster SDP methods. In  \cite{friedlander2016low}  it is noted that  the rank of primal PSD matrix variable is equal to the multiplicity of the matrix variable arising from the gauge dual formulation, and finding only those $r$ corresponding eigenvectors can recover the primal solution. In \cite{helmberg2000spectral}, a similar observation is made of the Lagrange dual variable  and thus the dual problem can be solved via a modified bundle method. More generally, the recently popularized conditional gradient algorithm (also called the Frank-Wolfe algorithm) efficiently solves norm-constrained problems for nonsymmetric matrices \cite{jaggi2010simple}, exploiting the fact that the dual norm minimizer can be computed efficiently; see also  \cite{candes2009exact, recht2010guaranteed, udell2016generalized}.

\paragraph{Nonconvex cases}
In close connection with these observations, \cite{burer2003nonlinear,burer2005local} proposed simply reformulating semidefinite matrix variables $Z = XX^T$, solving the ``standard" nonconvex SDP
\begin{equation}
\label{eq:reformulation}
%\begin{array}{ll}
\minimize{X\in \R^{n\times r}}\  \langle C, XX^T\rangle, \ \st \ \mA(XX^T) = b
%\end{array}
\end{equation}
by sequentially optimizing the Lagrangian.
%To solve \eqref{eq:main},~\cite{burer2003nonlinear} put forward an algorithm based on augmented Lagrangian method to find local optima of \eqref{eq:main}; unsurprisingly, due to the later observation of Boumel et. al, this method empirically found global optima almost all of the time.
However, solving \eqref{eq:main} is still numerically burdensome; in the augmented Lagrangian term, the objective is quartic in $R$, and is usually solved using an iterative numerical method, such as L-BFGS.

%The now-coined ``Burer-Monteiro method'' today is extended to a general matrix factorization approach  over the factor $X$, which if skinny significantly reduces memory and computational requirements.

\paragraph{Global optimality of a nonconvex problem with linear objective}
A main motivation behind solving rank-constrained problems using convex optimization methods come from key  results in
~\cite{pataki1998rank,barvinok1995problems} which show that for a linear SDP, when $X^*$ is the optimum and $r = \rank(X^*)$, then $\frac{r(r+1)}{2} \geq m$ where $m$ is the number of linear constraints.
Furthermore, a recent work~\cite{boumal2016non} shows that almost all local optima of FSDP are also global optima, suggesting that any stationary point of the FSDP is also a reasonable approximation of \eqref{eq:main},
if the constraint space of\eqref{eq:reformulation} is compact and sufficiently smooth, e.g.
$A_i Y$ linearly independent whenever $\langle A_i, YY^T\rangle = b_i$ for all $i = 1,\hdots, m$. The MAX-CUT problem satisfies this constraint; an example of a linear SDP without this condition is the phase retrieval problem \cite{candes2015phase}, when $m > n$.

\paragraph{Nonconvex constraint $\mC$}
Although there are many cases where the linear constraint in \eqref{eq:main} serves a distinct purpose, largely it is introduced to tighten the convex relaxation. When working in the nonconvex formulation, for many applications, the linear constraint becomes superfluous, and a more useful reformulation may be
\[
%\begin{array}{ll}
\minimize{x,y} \  g(x), \ \st\   x = y, y\in\mC,
%\end{array}
\]
for some nonconvex set $\mC$ (e.g. $\mC = \{-1,1\}^n$).
Note that the projection on $\mC$ is extremely easy, despite its nonconvexity.
Although less explored, this idea is not new; see \cite{boyd2011distributed} chapter 9.

\subsection{ADMM for nonconvex problems}
The alternating direction method of multipliers (ADMM) \cite{glowinski1975approximation,gabay1975dual} is a now popular method \cite{boyd2011distributed} for convex large-scale distributed optimization problems, with understood convergence rates \cite{eckstein2015understanding} and variations \cite{sun2015expected,yinthree,goldstein2014fast}. It is closely related to dual decomposition methods, but alternates its subproblems, and makes use of augmented Lagrangians, which smooths the subproblems and reduces the influence of the dual ascent step size. Although there are extensions to many variable blocks, most ADMM implementations use two variable block decompositions, solving
\[
%\label{e-admm}
%\begin{array}{ll}
\min_x \  g(x) + h(y), \ \st \  Ax = By
%\end{array}
\]
by alternatingly minimizing over each variable in the augmented Lagrangian
\[
 \mL_\rho (x,y;u) = g(x) + h(y) + u^T(Ax-By) + \frac{\rho}{2}\|Ax-By\|_2^2
\]
and then incrementally updating the dual variable:
%\begin{eqnarray*}
\[
x^+ = \arg\min_x\mL_\rho(x,y;u), y^+ = \arg\min_x\mL_\rho(x^+,y;u), u^+ = u + \rho(Ax^+-By^+).
\]
%\end{eqnarray*}
Here, any $\rho > 0$ will achieve convergence.

%\paragraph{Related work}
%Recently, the Alternating Direction Method of Multipliers (ADMM) \cite{gabay1976dual, glowinski1975approximation, lions1979splitting, eckstein1992douglas} has been widely applied and investigated in various fields; see the survey~\cite{boyd2011distributed} and the references therein.
%Specifically, it is favored because it is easy to apply to a distributed framework, and has been shown to converge quickly in practice.
%For convex optimization problems, ADMM can be shown to converge to a global optima in several ways: as a series of contractions of monotone operators \cite{eckstein1992douglas} or as the minimization of a global potential function \cite{boyd2011distributed}.
In general there is a lack of theoretical justification for ADMM on nonconvex problems  despite its good numerical performance.
Almost all works concerning ADMM on nonconvex problems investigate when nonconvexity is in the objective functions (\cite{hong2016convergence,wang2015global,li2015global,magnusson2016convergence,liu2017linearized}, and also
\cite{lu2017nonconvex,xu2012alternating} for matrix factorization)
Under a variety of assumptions (\eg~convergence or boundedness of dual objectives)  they are shown to convergence to a KKT stationary point.

In comparison, relatively fewer works deal with nonconvex constraints. ~\cite{jiang2014alternating} tackles polynomial optimization problems by minimizing a general objective over a spherical constraint $\|x\|_2 = 1$,  ~\cite{huang2016consensus} solves general QCQPs, and \cite{shen2014augmented} solves the low-rank-plus-sparse matrix separation problem.
In all cases, they show that all limit points are also KKT stationary points, but do not show that their algorithms will actually converge to the limit points. In this work, we investigate a class of nonconvex constrained problems, and show with much milder assumptions that the sequence always converges to a KKT stationary point.

\section{Linearized ADMM on full SDP}
%Scalable SDP solvers must exploit both low-rank and sparsity, in both memory and computation. We argue that using ADMM directly on the nonconvex formulation can more easily accomplish both goals.
We first investigate a reformulation of \eqref{eq:main} as
\begin{equation}
\label{eq:reform-1}
%\begin{array}{ll}
\minimize{Z,X,Y}  f(Z) + \delta_{\{0\}}(\mA(Z) - b) + \delta_{\mC}(Y), \st  Z = (XY^T)_\Omega, X = Y\\
%\end{array}
\end{equation}
with variables $Z\in \symm^{n\times n}$, $X\in \R^{n\times r}$, and  $Y\in \R^{n\times r}$. The affine and $\mC$ constraints are lifted to the objective via an indicator function
\[
\delta_{\mC}(x) =
\begin{cases}
0 & \text{ if } x\in \mC,\\
\infty & \text{ else.}
\end{cases}
\]
The notation $A_{\Omega}$ for a symmetric matrix $A$ is the projection of $A$ on the sparsity pattern $\Omega$:
\[
(A_{\Omega})_{ij} =
\begin{cases}
A_{ij}, & \text{ if } \{i,j\}\in \Omega\\
0, & \text{ else,}
\end{cases}
\]
and we write $A\in \symm_{\Omega}^n$ if $A_{\Omega} = A$.
Specifically, $\Omega$ captures the \emph{effective sparsity} of the problem;
that is, $f(Z) = f(Z_{\Omega})$ and $\mA(Z) = \mA(Z_{\Omega})$.
We assume $\{i,i\}\in \Omega$ for all $i$, so the second is trivially true.

\paragraph{Duality}
As shown in
\cite{rockafellar1974augmented}, a notion of a dual problem can be established via the augmented Lagrangian of \eqref{eq:reform-1}
\begin{eqnarray}
\mL_\rho(Z,X,Y;S,U) &=&  f(Z)  + \delta_{\mC}(Y)  + \langle U,X-Y\rangle
 +\langle S,Z - XY^T\rangle \nonumber\\
 &&\quad + \frac{\rho}{2}\|X-Y\|_F^2
  + \frac{\rho}{2}\|Z - XY^T\|_F^2\label{eq:AugLag-1}
\end{eqnarray}
where the dual problem is
$
\maximize{S,U} \min_{Z,X,Y}\mL_\rho(Z,X,Y;S,U).
$
The minimization of $\mL_\rho$ over $Z$ and $X$ is the solution to
\begin{equation}
\label{eq:optimalXZ}
\begin{array}{rcl}
\nabla f(Z) - \mA^*(\nu) + S + \rho(Z-XY^T) &=& 0\\
U -S Y + \rho(XY^TY-ZY) + \rho(X-Y) &=& 0\\
\mA(Z) &=& b
\end{array}
\end{equation}
where $\nu > 0$ is a Lagrange dual variable for the local constraint $\mA(Z) = b$.
The minimization of $\mL_\rho$ over  $Y$ is the solution to the generalized projection problem
\begin{equation}
\label{eq:optimalY}
\min_{Y\in \mC} \quad \langle Y - \hat Y, Y-\hat Y\rangle_H = \tr(( Y - \hat Y)H ( Y-\hat Y)^T)
\end{equation}
where
$
\hat Y = U + SX + \rho (X+Z^TX),\quad
H = \rho (I+X^TX).
$
For general nonconvex problems, it is difficult to guarantee global minimality.
Here we introduce two sought-after properties that are more reasonably attainable.
\begin{definition}\cite{clarke1990optimization}
The \emph{tangent cone} of a nonconvex set $\mC$ at $x$ is  given by
\[
\mathcal T_\mC(x) = \{d : \text{ for all } t\to 0, \hat x \to x, \hat x \in \mC, \text{ there exists } \hat d \to d, \hat x +t \hat d\in \mC\}.
\]
The \emph{normal cone} of $\mC$  at $x$ ($:= \mN_\mC(x)$) is the polar of the tangent cone.
\end{definition}
\begin{definition}
For a minimization of a smooth constrained function $\underset{x\in \mC}{\min}\;f(x)$  we say that $x^*$ is a \emph{KKT-stationary point} if $-\nabla f(x^*)\in \mN_\mC(x^*)$.
\label{def:stationary}
\end{definition}

\begin{definition}
For a function defined over $M$ variables $\mL(X_1,\hdots, X_m)$, we say that $X_1^*,\ldots,X_m^*$ are \emph{(block) coordinatewise minimum points} if for each $k = 1,\hdots, m$,
$
X_k^* = \argmin{X}\;\mL(X_1^*,\ldots,X_{k-1}^*,X,X_{k+1}^*,\hdots, X_m^*).
$
\label{def:coordinate}
\end{definition}
Note that it is not always the case that stationarity is stronger than coordinatewise minimum. A simple example is $\mC = \{-1,1\}^n$. Then for all points $x\in \mC$, the tangent cone is $\{0\}$ and the normal cone is $\R^n$. Then  every point in $\mC$ is stationary, no matter what the objective function.

\begin{proposition}
If Alg. \ref{a:matrix} converges to  coordinatewise minimum points\\
 $((X,Z)^*, Y^*, S^*, U^*)$, then the primal points 
 i)  satisfy \eqref{eq:optimalXZ} for some choice of $\nu \geq 0$,
ii)  minimize \eqref{eq:optimalY},
iii) and are primal-feasible, e.g. $X^* = Y^*$ and $(X^*(Y^T)^*)_\Omega = Z^*$.
Furthermore, $(X^*, Y^*, Z^*, S^*, U^*)$ are stationary points of \eqref{eq:AugLag-1}
\end{proposition}
\begin{proof}
It is clear that the convergent points of Alg. \ref{a:matrix} exactly satisfy the three conditions.
To show that these points are stationary, note that the augmented Lagrangian is convex with respect to $X,Z$ jointly, and is a projection on a compact set $\mC$ with respect to $Y$. Therefore
\[
\nabla_{X,Z,S,U} \mL_\rho(Z^*,X^*,Y^*;S^*,U^*) = 0, \qquad -\nabla_Y \bar \mL_\rho(Z^*,X^*,Y^*;S^*,U^*) \in \mN_\mC(Y^*)
\]
where
$
\bar \mL_\rho(Z,X,Y;S,U) =  - \langle U,Y\rangle
 -\langle S, XY^T\rangle + \frac{\rho}{2}\|X-Y\|_F^2
  + \frac{\rho}{2}\|Z - XY^T\|_F^2
$
 with all the differentiable terms of $\mL_\rho$ involving $Y$.
\end{proof}

\subsection{Linearized ADMM}
We propose to solve \eqref{eq:reform-1} via the \emph{linearized} ADMM, e.g. where at each iteration, the objective is replaced by its current linearization
\[
f(Z) \approx \hat f^k(Z) := f(Z^{k-1}) + \langle \nabla f(Z^{k-1}), Z-Z^{k-1} \rangle.
\]
We then build the linearized augmented Lagrangian function as
\begin{eqnarray}
\hat \mL^k(Z,X,Y;S,U) &=&  g_k(X,Z) +  h(Y)  + \langle U,X-Y\rangle
 +\langle S,Z - XY^T\rangle +\nonumber \\
 &&\quad \frac{\rho}{2}\|X-Y\|_F^2
  + \frac{\rho}{2}\|Z - XY^T\|_F^2\label{eq:AugLag-lin}
\end{eqnarray}
where
$
g_k(X,Z) = \hat f^k(Z) + \delta_{\{0\}} (\mA(Z)-b) ,\quad h(Y) = \delta_{\mC}(Y)
$
and $S\in\reals^{n\times n}$ and $U\in\reals^{n\times r}$ are the dual variables corresponding to the two coupling constraints.
The full algorithm is given in Alg. \ref{a:matrix}.

\begin{algorithm}
	\caption{ADMM for solving \eqref{eq:reform-1}}
	\label{a:matrix}
	\begin{algorithmic}[1]
		\Inputs{$\rho_0>0$, $\alpha>1$, tol $\epsilon > 0$}
		\Initialize{$Z^0, X^0;S^0,U^0$ as random matrices}
		\Outputs {$Z$, $X=Y$}
		\For{$k = 1 \hdots$}
		\State Update ${Y^{k+1}}$ the solution of
			\begin{equation}\label{eq:xnewXxy}
			%\begin{array}{ll}
			\minimize{Y\in \R^{n\times k}}  \|Z^{k} - X^{k}Y^T + \frac{S^k}{\rho^{k}}\|_F^2
			+ \|X^{k}-Y + \frac{U^k}{\rho^k}\|_F^2, \st  Y\in \mC
			%\end{array}
			\end{equation}
		\State Update $(Z,X)^{k+1}$ as the solutions of
			\begin{equation}\label{eq:z1z2Xxy}
			%\begin{array}{ll}
			\underset{X,Z\in\symm_{\Omega}^n}{\min}  \mL_{k+1}(Z,X,Y^{k+1};S^{k},U^{k};\rho^k), \st   \mA(Z) = b
			%\end{array}
			\end{equation}
			where $\mL$ is the linearized augmented Lagrangian as defined in \eqref{eq:AugLag-lin}.
		\State Update $S, U$ and $\rho$ via
			\bea\label{eq:LambdaUpdXxy}
			S^{k+1} &=& S^{k} + \rho^k({Z}^{k+1} - {X}^{k+1}({Y}^{k+1})^T)_\Omega\nonumber\\
			U^{k+1} &=& U^{k} + \rho^k({X}^{k+1} - {Y}^{k+1})\nonumber\\
			\rho^{k+1} &=& \alpha\rho^{k}
			\eea
		\If{$\max \{\|X^{k}-Y^k\| ,\|(Z^{k}-X^k(Y^k)^T)_{\Omega}\|\} \leq \epsilon$}
		\State {\textbf{break}}
		\EndIf
		\EndFor
	\end{algorithmic}
\end{algorithm}

\paragraph{Minimizing over $Y$}
The generalized projection \eqref{eq:optimalY} can be solved a number of ways.
Note if $r = 1$ then $H$ is a positive scalar, and the problem reduces to $Y^+ = \proj_{Y\in \mC} \left(\frac{1}{H}\hat Y\right)$.
When $\mC = \{-1,1\}^n$, this process reduces to recovering the signs of  $\hat Y$ i.e., $Y_i = \textbf{sign}_{\mC}(\hat Y_i)$,
 and when $\mC = \{u:\|u\|_2 = 1\}$ the set of unit-norm vectors, $Y$ is just a properly scaled version of $\hat Y$: $Y = \frac{1}{\|\hat Y\|_2}\hat Y.$
 However, in general, it is difficult to compute the generalized projection over a nonconvex set.
 When $\mC$ is convex, the generalized projection problem \eqref{eq:optimalY} can be computed using projected gradient descent. Note that the objective of \eqref{eq:optimalY} is 1-strongly convex; thus we expect fast convergence in this subproblem. In practice, we find that if $r$ is not too large, often a few tens of iterations is enough.

\paragraph{Minimizing over $X$ and $Z$.}
Using standard linear algebra techniques, the linear system \eqref{eq:optimalXZ} can be reduced to a few simple instructions. First, we solve for the Lagrange dual variable $\nu$ associated with the linear constraints (and localized to the minimization of $X$ and $Z$):
\begin{equation}
\label{eq:solvenu}
\mA(\mA^*(\nu)(YY^T + I)) = \rho(b-\mA(DY^T+YY^T)) + \mA((G+S)(I + YY^T)),
\end{equation}
where $D =  \frac{1}{\rho}(S Y -U) + Y$ and $G = \nabla f(Z^{k-1})$ the local gradient estimate.
When $\mA = \diag$, \eqref{eq:solvenu} reduces to $n$ scalar element-wise computations \\
$
\nu_i = \frac{ \rho(b-(DY^T)_{ii}) + ((G+S)(I + YY^T))_{ii}}{(YY^T)_{ii}+1}.
$
When $\mA = \tr$,
$
\nu = \frac{\rho(b-\tr(DY^T)+\tr((G+S)(I + YY^T))}{\tr(YY^T) + 1}.
$
Note that in both cases, no $n\times n$ matrix need ever be formed, so the memory requirement remains $O(nr)$. (See appendix for elaboration.)
Then the primal variables are recovered via
$
X = BY + D, \quad \text{ and } \quad
Z = (XY^T)_\Omega + B,
$
with
$
B = -\frac{1}{\rho}(C-\mA^*(\nu) + S).
$
In these cases, the complexity is dominated by multiplications between $n\times n$ and $n\times r$ matrices. Thus, the method is especially efficient when $r \ll n$.

\subsection{Convergence analysis}

%In this section, we aim to establish the global convergence of Algorithms 1 to a stationary point of \eqref{eq:FormulationXxy1} under some conditions~.

%\begin{Assumption}
%\begin{enumerate}
%\item
%The sequences $\{f(Z^k)\}$ and $\{S^k,U^k\}$ are bounded in norm.
%\item
%$f(Z)$ has bounded non-convexity. Mathematically, for any $Z_1,Z_2\in\symm^n$, it holds that $f(Z_1) \ge f(Z_2) + \langle \nabla f(Z_2), Z_1-Z_2\rangle - \frac{L_f}{2}\|Z_1-Z_2\|_F^2$ with $L_f < +\infty$. \red{Is this ok, even with $H_g$?}
%\end{enumerate}
%\end{Assumption}

\begin{theorem}
Assume that $f(Z)$ is $L_f$-smooth.
Assume the dual variables are bounded,
 e.g.
$
\max\{\|S^k\|_F,\|U^k\|_F,\|Y^k\|_F\}_k \leq B_P < +\infty,
$
and
$\frac{L_f}{\sigma_{\max}}$ is bounded above, where
$
\sigma_{\max} = 1 - \frac{\sqrt{\sigma_Y^4 + 4\sigma_Y^2} - \sigma_Y^2}{2},\quad \sigma_Y = \|Y^{k+1}\|_2.
$
Then by running  Alg. \ref{a:matrix} with $\rho^k = \alpha \rho^{k-1}=\alpha^k\rho_0$, if $\mL_k$ is bounded below, then
  the sequence  $\{P^k,D^k\}$ converges to a stationary point of \eqref{eq:AugLag-1}.
\end{theorem}

\begin{proof}
  See section \ref{sec:matrixproof} in the appendix.
\end{proof}

%If we pick $r = 1$, then $X$ and $Y$ are length-$n$ binary vectors, and $x = X$ is a locally optimal solution of \eqref{eq:main}. If $\frac{r(r+1)}{2} > n$ (e.g, $\ceil{\sqrt{2n}}$), $X$ can be projected on a rank-1 solution using a randomization technique popularized by Goemans/Williamson (\cite{goemans1995improved}). For max-cut problem, such randomization technique can generate an approximately optimal cut within a ratio of $0.878$.

\begin{corollary}
	If $r \geq \ceil{\sqrt{2n}}$ and the  stationary point of Algorithm \ref{a:matrix} converges to a second order critical point of \eqref{eq:main}, then it is globally optimal for the convex relaxation of \eqref{eq:reformulation}~\cite{boumal2016non}.
\end{corollary}
%\begin{proof}
%  Given the search space of the SDR \eqref{eq:FormulationXxy0SDR} is compact and that of \eqref{eq:FormulationXxy1} is a smooth manifold, combining Theorem \ref{lemma:LagDecentXxy} and Theorem 2 in~\cite{boumal2016non} will trivially result to the conclusion.
%\end{proof}

Unfortunately, the extension of KKT stationary points to global minima is not yet known when $\frac{r(r+1)}{2} < n$ (\ie, $r = 1$). However, our empirical results suggest that even when $r = 1$, often a local solution to \eqref{eq:reformulation} well-approximates the global solution to \eqref{eq:main}.

\section{ADMM on simplified nonconvex SDP}
When the linear constraints are not present, \eqref{eq:main} can be reformulated without $Z$, into
\begin{equation}
\label{eq:reform-2}
%\begin{array}{ll}
\minimize{X,Y}  g(X) + \delta_{\mC}(Y), \st  X = Y
%\end{array}
\end{equation}
with matrix variables $X\in \R^{n\times r},Y\in\R^{n\times r}$, and
where $g(X) = f(XX^T)$ is smooth.
We can also define an augmented Lagrangian of \eqref{eq:reform-2} as
$
\mL_\rho(X,Y;U) = g(X) + \delta_{\mC}(Y) + \langle U,X-Y\rangle + \frac{\rho}{2}\|X-Y\|_F^2.
$

\begin{theorem}
The coordinatewise minimum points $X^*=Y^*$ satisfying
\begin{equation}
\label{eq:coordinate-2}
\begin{array}{rcl}
0 &=& \nabla g(X^*) + U + \rho(X-Y) \\
Y &=& \displaystyle \proj_{\mC}(X+\frac{1}{\rho}U)\\
X &=& Y
\end{array}
\end{equation}
are the stationary points of
the problem
\begin{equation}
\label{eq:reform-simp}
%\begin{array}{ll}
\minimize{X}  g(X), \st  X\in \mC.
%\end{array}
\end{equation}
\end{theorem}

\begin{proof}
The KKT stationary points of \eqref{eq:reform-simp} can be characterized in terms of the normal cone of $\mC$ at $X^*$; specifically, $X^*$ is stationary iff
\[
\langle \nabla g(X^*),X-X^*\rangle \leq 0,\quad \forall X\in \mC\cap\mN_\epsilon(X^*)
\]
where $\mN_\epsilon(X^*)$ is some small neighborhood containing $X^*$.
(This is an equivalent definition of the Clarke stationary point\cite{clarke1990optimization}, since in a close enough neighbourhood to $X^*$, the subdifferential of $\delta_{\mC}(x)$ is $\mN_{\mC}(x)$.)

Combining terms in \eqref{eq:coordinate-2} gives $X^*=Y^*$ satisfying
$
X^* = \proj_{\mC}\left(X^*-\frac{1}{\rho}\nabla g(X^*)\right).
$
The optimality condition of the projection is
$
\langle X - (X - \frac{1}{\rho} \nabla g(X^*)) ,X-X^*\rangle \leq 0,\quad \forall X\in \mC\cap\mN_\epsilon(X^*)
$
which reduces to the desired condition.
\end{proof}

\begin{algorithm}
	\caption{ADMM for solving \eqref{eq:reform-simp}}
	\begin{algorithmic}[1]
		\Inputs{$\rho_0>0$, $\alpha>1$, tol $\epsilon > 0$}
		\Initialize{$Z^0, X^0;S^0,U^0$ as random matrices}
		\Outputs {$Z$, $X=Y$}
		\For{$k = 1 \hdots$}
		\State Update ${Y^{k+1}}$ the solution of
			\begin{equation}
			%\begin{array}{ll}
			\minimize{Y\in R^{n\times k}}   \|X^{k}-Y + \frac{U^k}{\rho^k}\|_F^2, 	\st  Y\in \mC.
			%\end{array}
			\end{equation}
		\State Update $X^{k+1}$ as the solution of
			\begin{equation}
			0 = \nabla g(X) + U + \rho(X-Y).
			\end{equation}
		\State Update $U$ and $\rho$ via
			\bea
			U^{k+1} = & U^{k} + \rho^k({X}^{k+1} - {Y}^{k+1}), \rho^{k+1} = \alpha\rho^{k}.
			\eea
		\If{$\|X^{k}-Y^k\|_F \leq \epsilon$}
		\State {\textbf{break}}
		\EndIf
		\EndFor
	\end{algorithmic}
\label{a:vector}
\end{algorithm}

\subsection{ADMM}
The alternating steps in minimizing the augmented Lagrangian over the primal variables
are extremely simple, compared with the previous matrix formulation.
In general we are considering $f(X)$ linear (in which case the update of $X$ involves only addition) or quadratic with strictly positive diagonal Hessian (which adds a small scaling step).
$
\mC = \{-1,1\}^n,\qquad\mC = \{x : \|x\|_2 = 1\},
$ even when $r > 1$.

\subsection{Convergence analysis}

\begin{definition}\label{def:g}
A differentiable convex function $g(X)$
is $L_g$-smooth
and $H_g$-strongly convex
over $\R^n$ if
for any $X$, $Y$,
$
g(X)-g(Y) \geq \langle \nabla f(X),X-Y\rangle - \frac{L_g}{2}\|X-Y\|_F^2
$
and
$
g(X)-g(Y) \leq \langle\nabla f(X),X-Y\rangle - \frac{H_g}{2}\|X-Y\|_F^2.
$
\end{definition}

\begin{theorem}\label{thm:simple-convergence}
Assume $g(X)$ is lower bounded over $\mC$, and is $L_g$-smooth.
Given a sequence $\{\rho^k\}$ such that
\[
\frac{\rho^k-3L_g}{2} - L_g^2\frac{\rho^{k+1}+\rho^{k}}{2(\rho^k)^2} > 0, \qquad \rho^k > L_g
\]
for all $k$, then
under Algorithm \ref{a:vector} the augmented Lagrangian  $\mL(X^{k},Y^k;U^k)$ is lower bounded and convergent, with
 $\{X^k,Y^k,U^k\}\to \{X^*,Y^*,U^*\}$ a stationary and feasible solution of \eqref{eq:reform-simp}.
\end{theorem}

\begin{proof}
See section \ref{sec:vectorconvergence} in the appendix.
\end{proof}

\paragraph{Remark}: Convergence is  guaranteed under a constant penalty coefficient
$
\rho_k \equiv \rho^0  \geq \frac{3+\sqrt{17}}{2} L_g,\qquad \alpha = 1.
$
% can also guarantee the convergence. I
However, in implementation, we find empirically that increasing $\{\rho^k\}$ from a relatively small $\rho^0$  can encourage convergence to more useful global minima.

\begin{theorem}\label{thm:LagDecent-linear}
If $g(X)$ is $H_g$-strongly convex and  $\rho^k = \rho$ constant, with
$
\frac{\rho+H_g}{2}\geq \frac{L_g^2}{\rho},\quad \rho > L_g
$
then
under Algorithm \ref{a:vector} the augmented Lagrangian  $\mL(X^{k},Y^k;U^k)$ converges to \\ $\mL(X^*,Y^*,U^*)$ at a linear rate.
\end{theorem}
\begin{proof}
See section \ref{sec:vectorconvergence-linear} in the appendix.
\end{proof}

\section{Numerical experiments}
In this section, we give numerical results on the proposed methods for community detection, MAX-CUT, image segmentation, and symmetric matrix factorization.
In each application, we evaluate and compare these four methods.
%\begin{enumerate}
	i) SD: the solution to a  semidefinite relaxation of \eqref{eq:main} (SDR), where $\mC = \R^{n,r}$.
	The binary vector factor $x$ where $xx^T = Z$ is  is recovered   using a  Goemans-Williamson style rounding. \cite{goemans1995improved} technique. This is our baseline method, and is described in more detail below.
	ii) MR1: Algorithm \ref{a:matrix} with $r = 1$.
	iii) MRR: Algorithm \ref{a:matrix} with $r = \ceil{\sqrt{2n}}$, then rounded to a binary vector using  a nonsymmetric version of the Goemans-Williamson style rounding \cite{goemans1995improved} technique.
	Both MR1 and MRR have the following stopping criterion $\max\{P^{(k)},D^{(k)}\} \leq  \epsilon$ for some tolerance parameter $\epsilon > 0$, where:
$	
%\begin{eqnarray*}
	P^{(k)} := \left\{ \frac{\|Z^k-Z^{k-1}\|_2}{\|Z^k\|_2},\frac{\|X^k-X^{k-1}\|_2}{\|X^k\|_2},\frac{\|Y^k-Y^{k-1}\|_2}{\|Y^k\|_2}\right\},
	D^{(k)} := \max\left\{\frac{\|Z^{(k)} - X^{(k)}(Y^{(k)})^T\|_2}{\|Z^k\|_2},\frac{\|X^{(k)}-Y^{(k)}\|_2}{\|X^{(k)}\|_2}\right\}.
	%\end{eqnarray*}
$
	(Here $D^{(k)}$ is also proportional to the difference in dual iterates, and thus $P^{(k)}$ and $D^{(k)}$ can be interpreted as primal and dual residuals, respectively.)

	iv) V: Algorithm \ref{a:vector}, with stopping criterion $\max\{P^{(k)},D^{(k)}\} \leq  \epsilon$ where	
$
	P^{(k)} := \left\{\frac{\|x^k-x^{k-1}\|_2}{\|x^k\|_2}, \frac{\|y^k-y^{k-1}\|_2}{\|y^k\|_2}\right\}, \qquad
	D^{(k)} := \frac{\|x^k-y^k\|_2}{\|x^k\|_2}.
$
The same primal and dual residual interpretation can be used here as well.
%\end{enumerate}
In all cases, we use the following scheme for $\rho$:
$
 \rho^k = \min\{\rho_{\max},\rho^{k-1}*\gamma\}.
$
where $\rho_{\max}\approx 10,000$ and $\gamma\approx 1.05$ (slightly larger than 1).

\paragraph{Solving the baseline (SDR)}
As a baseline, we compare against the solution of the semidefinite relaxed problem without factor variables $X$ (e.g. $\mC = \R^{n,n}$):
\begin{equation}
%\begin{array}{ll}
\minimize{Z}  f(Z), \st  \mA(Z) = b, Z \succeq 0.
%\end{array}
\label{eq:sdr}
\end{equation}	
For a fair comparison, we use a first-order splitting method very similar to ADMM, which is the Dougals-Rachford Splitting (DRS) method (\cite{lions1979splitting,douglas1956numerical}, see also \cite{spingarn1985applications,eckstein1992douglas}). We introduce dummy variables and solve the reformulation of \eqref{eq:sdr}
\[
%\begin{array}{ll}
\minimize{Z_1,Z_2,Z_3}  g_1(Z_1)+g_2(Z_2) + g_3(Z_3), \st  Z_1+Z_2+Z_3\\
%\end{array}
\]
where
$
g_1(Z_1) = \tr(CZ_1),  g_2(Z_2) =
\begin{cases}
0, &\mA(Z_2) = b\\
+\infty, & \text{ else, and}
\end{cases}
 g_3(Z_3) =
\begin{cases}
0, & Z_3\succeq 0 \\
+\infty, & \text{ else.}
\end{cases}
$
An application of the DRS on this reformulation (see also Alg. 3.1 in \cite{combettes2008proximal}) is then the following iteration scheme: for $i = 1,2,3$,
\begin{eqnarray*}
X_i^{(k+1)} &=& \prox_{t g_i}(Z_i), \hat Y_i = 2X_i^{(k+1)} -Z_i^{(k)},\\
Y^{(k+1)} &=& \frac{1}{3}( X_1^{(k+1)} +  X_2^{(k+1)} +  X_3^{(k+1)}),\\
Z_i^{(k+1)}  &=& Z_i^{(k)}  + \rho(Y^{(k+1)} -X_i^{(k+1)} )
\end{eqnarray*}
and for a convex function $f$
$
z = \prox_{tf}(u) \iff \argmin{z} \; f(z) + \frac{1}{2t} \|z-u\|_2^2.
$

\paragraph{Rounding} Following the technique in \cite{goemans1995improved}, we can estimate $x$ from a rank $r$ matrix $X \approx xx^T$ by randomly projecting the main eigenspaces on the unit sphere. The exact procedure is as follows.
%\begin{itemize}
i) For the symmetric SDP solution $X$, we first do an eigenvalue decomposition $X = Q\Lambda Q^T$ and form  a factor $F = Q\Lambda^{1/2}$ where the diagonal elements of $\Lambda$ are in decreasing magnitude order. Then we scan $k = 1,\hdots, n$ and find $x_{k, t} = \sign(F_k z_t)$ for trials $t = 1,\hdots  10$. Here, $F_k$ contain the first $k$ columns of $F$, and each element of $z_t\in \R^k$ is drawn i.i.d from a normal Gaussian distribution.
We report the values for $x_r = \argmin{x_{k,t}}\;\{x_r^TCx_r\}$.
ii) For the MRR method, we repeat the procedure using a factor $F = U\Sigma^{1/2}$ where $X=U\Sigma V^T$ is the SVD of $X$.
iii) For MR1 and V, we simply take $x_r = \sign(x)$  as the binary solution.
%\end{itemize}

\paragraph{Computer information}
The following simulations are performed on a standard desktop computer with an
Intel Xeon processor (3.6 GHz), and 32 GB of RAM . It is running with Matlab R2017a.

\paragraph{MAX-CUT}
Table \ref{t:maxcut-perf} gives the best MAX-CUT values using best-of-random-guesses and our approaches over four examples from the 7th DIMACS Implementation Challenge in 2002.\footnote{See \texttt{http://dimacs.rutgers.edu/Workshops/7thchallenge/.} Problems downloaded from \texttt{http://www.optsicom.es/maxcut/}}
Often, we find the quality of our recovered  solutions   close to the best-known solutions, and often achieve similar suboptimality as the rounded SDR solutions. However, the runtime comparison (Fig. \ref{f:maxcut-runtime}) suggests the ADMM methods (especially MR1 and SDR) are much more computationally efficient and scalable.
All experiments are performed with $\epsilon = 1\times 10^{-3}$.

	\begin{table}
		\begin{center}
			{\footnotesize
            \begin{tabular}{ c|cc|c|cccc  }
\hline\hline
		database		&		n		&		sparsity		&		BK		&	V		&		MR1	& MRR & SDR	\\
\hline
g3-8	&	512	&	0.012	&	{\scriptsize 41684814}	&	{\scriptsize 34105231}	&	{\scriptsize 36780180}	&	{\scriptsize 35943350}	&	{\scriptsize 33424095}\\
g3-15	&	3375	&	0.018	&	{\scriptsize 281029888}	&	{\scriptsize 235893612}	&	{\scriptsize 255681256}	&	{\scriptsize 241740931}	&	{\scriptsize 212669181}\\
pm3-8-50	&	512	&	0.012	&	454	&	394	&	346	&	378	&	416\\
pm3-15-50	&	3375	&	0.018	&	2964	&	2594	&	1966	&	2140	&	2616\\
G1	&	800	&	0.0599	&	11624	&	10938	&	11047	&	11321	&	11360\\
G2	&	800	&	0.0599	&	11620	&	10834	&	11082	&	11144	&	11343\\
G3	&	800	&	0.0599	&	11622	&	10858	&	10894	&	11174	&	11367\\
G4	&	800	&	0.0599	&	11646	&	10849	&	10760	&	11192	&	11429\\
G5	&	800	&	0.0599	&	11631	&	10796	&	10783	&	11352	&	11394\\
G6	&	800	&	0.0599	&	2178	&	1853	&	1820	&	1949	&	1941\\
G7	&	800	&	0.0599	&	2003	&	1694	&	1644	&	1705	&	1774\\
G8	&	800	&	0.0599	&	2003	&	1688	&	1641	&	1728	&	1766\\
G9	&	800	&	0.0599	&	2048	&	1771	&	1681	&	1807	&	1830\\
G10	&	800	&	0.0599	&	1994	&	1662	&	1641	&	1737	&	1732\\
G11	&	800	&	0.005	&	564	&	496	&	460	&	480	&	506\\
G12	&	800	&	0.005	&	556	&	486	&	448	&	480	&	512\\
G13	&	800	&	0.005	&	580	&	516	&	476	&	498	&	528\\
G14	&	800	&	0.0147	&	3060	&	2715	&	2768	&	2861	&	2901\\
G15	&	800	&	0.0146	&	3049	&	2625	&	2810	&	2803	&	2884\\
G16	&	800	&	0.0146	&	3045	&	2667	&	2736	&	2862	&	2910\\
G17	&	800	&	0.0146	&	3043	&	2638	&	2789	&	2840	&	2920\\
G18	&	800	&	0.0147	&	988	&	798	&	768	&	841	&	858\\
G19	&	800	&	0.0146	&	903	&	700	&	641	&	694	&	780\\
G20	&	800	&	0.0146	&	941	&	723	&	691	&	766	&	788\\
G21	&	800	&	0.0146	&	931	&	696	&	713	&	810	&	794\\
G22	&	2000	&	0.01	&	13346	&	12461	&	12548	&	12751	&	12926\\
G23	&	2000	&	0.01	&	13317	&	12540	&	12528	&	12853	&	12889\\
G24	&	2000	&	0.01	&	13314	&	12540	&	12447	&	12723	&	12904\\
G25	&	2000	&	0.01	&	13326	&	12447	&	12558	&	12733	&	12874\\
G26	&	2000	&	0.01	&	13314	&	12445	&	12475	&	12718	&	12847\\
G27	&	2000	&	0.01	&	3318	&	2824	&	2508	&	2807	&	2909\\
G28	&	2000	&	0.01	&	3285	&	2753	&	2518	&	2796	&	2845\\
G29	&	2000	&	0.01	&	3389	&	2864	&	2628	&	2901	&	2896\\
G30	&	2000	&	0.01	&	3403	&	2887	&	2639	&	2937	&	2971\\
G31	&	2000	&	0.01	&	3288	&	2833	&	2518	&	2902	&	2825\\
G32	&	2000	&	0.002	&	1398	&	1220	&	1066	&	1204	&	1254\\
G33	&	2000	&	0.002	&	1376	&	1202	&	1054	&	1166	&	1250\\
G34	&	2000	&	0.002	&	1372	&	1208	&	1096	&	1170	&	1222\\
G35	&	2000	&	0.0059	&	7670	&	6605	&	6914	&	6764	&	7209\\
G36	&	2000	&	0.0059	&	7660	&	6564	&	6943	&	6598	&	7228\\
G37	&	2000	&	0.0059	&	7666	&	6478	&	6839	&	6789	&	7183\\
G38	&	2000	&	0.0059	&	7681	&	6486	&	6759	&	6768	&	7212\\
G39	&	2000	&	0.0059	&	2395	&	1616	&	1697	&	1840	&	1997\\
G40	&	2000	&	0.0059	&	2387	&	1617	&	1438	&	1921	&	1890\\
G41	&	2000	&	0.0059	&	2398	&	1606	&	1656	&	1778	&	1899\\
G42	&	2000	&	0.0059	&	2469	&	1707	&	1756	&	1862	&	1971\\
G43	&	1000	&	0.02	&	6659	&	6222	&	6236	&	6398	&	6475\\
G44	&	1000	&	0.02	&	6648	&	6275	&	6192	&	6447	&	6458\\
G45	&	1000	&	0.02	&	6652	&	6243	&	6255	&	6407	&	6454\\
G46	&	1000	&	0.02	&	6645	&	6217	&	6233	&	6398	&	6407\\
G47	&	1000	&	0.02	&	6656	&	6221	&	6266	&	6433	&	6454\\
G48	&	3000	&	0.0013	&	6000	&	5882	&	5006	&	5402	&	6000\\
G49	&	3000	&	0.0013	&	6000	&	5844	&	5038	&	5362	&	6000\\
G50	&	3000	&	0.0013	&	5880	&	5814	&	4994	&	5410	&	5880\\
G51	&	1000	&	0.0118	&	3846	&	3317	&	3446	&	3524	&	3642\\
G52	&	1000	&	0.0118	&	3849	&	3360	&	3471	&	3499	&	3662\\
G53	&	1000	&	0.0118	&	3846	&	3323	&	3510	&	3516	&	3660\\
G54	&	1000	&	0.0118	&	3846	&	3306	&	3428	&	3509	&	3651\\
\hline\hline
			\end{tabular}}
			\caption{MAX-CUT values for graphs from the 7th DIMACS Challenge. MRR = matrix formulation, $r = \ceil{\sqrt{2n}}$. SDR = SDP relaxation + rounding technique.}			
			\label{t:maxcut-perf}
		\end{center}
\end{table}

\begin{figure}
\begin{center}
	\includegraphics[scale=0.75]{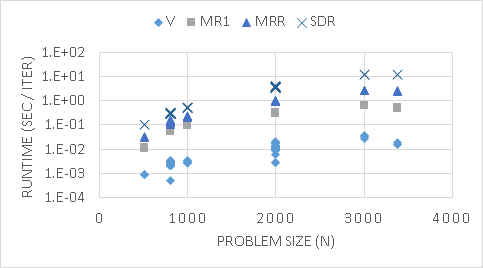}
	\includegraphics[scale=0.75]{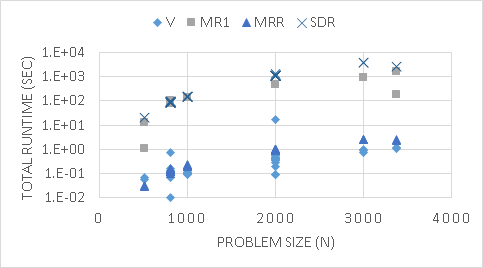}
	\end{center}
	\caption{Time comparisons for DIMACS problems. Left: average runtime per iteration. Right: total runtime. We observe that both V and MRR converge in relatively few number of iterations, with MR1 taking slightly longer. However, as previously observed with splitting methods, the convergence rate is sensitive to the parameter choices $\rho^{(t)}$. For best performance, we start with a relatively small initial penalty coefficient and increase it with the iteration until the upper bound is achieved.
	%Also, the eigenvalue information of the objective Hessian can be used when the problem size is relative small.
	}
	\label{f:maxcut-runtime}
\end{figure}

\paragraph{Image segmentation}
Both community detection and MAX-CUT can be used in image segmentation, where each pixel is a node and the similarity between pixels form the weight of the edges. Generally, solving \eqref{eq:main} for this application is not preferred, since the number of pixels in even a moderately sized image is extremely large. However, because of our fast methods, we successfully performed image segmentation on several thumbnail-sized images, in figure \ref{f:imageseg}.

The $C$ matrix is composed as follows. For each pixel, we compose two feature vectors: $f_c^{ij}$ containing the RGB values and $f_p^{ij}$ containing the pixel location. Scaling $f_c^{ij}$ by some weight $c$, we form the concatenated feature vector $f^{ij} = [f_c^{ij}, cf_p^{ij}]$, and form the weighted adjacency matrix as the squared distance matrix between each feature vector $A_{(ij), (kl)} = \|f^{ij} - f^{kl}\|_2^2$.
For MAX-CUT, we again form $C = A - \Diag(A\mb 1)$ as before. For community detection, since we do not have exact $p$ and $q$ values, we use an approximation as $C = a\mb1 \mb1^T-A$ where $a = \frac{1}{n^2}\mb1^TA\mb1$ the mean value of $A$. Sweeping $C$ and $\rho_0$, we give the best qualitative result in figure \ref{f:imageseg}.

\begin{figure}
\begin{center}
	\includegraphics[width=3in]{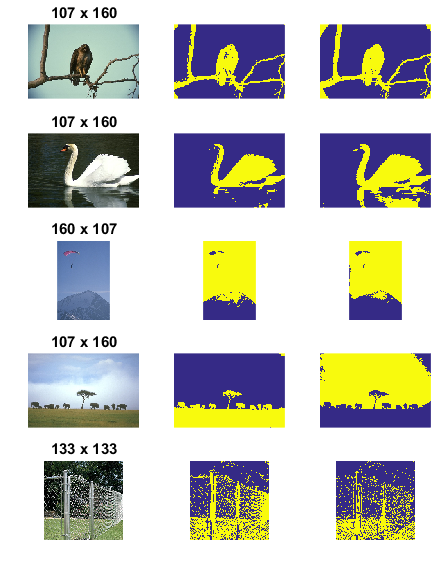}
\end{center}
	\caption{Image segmentation. The center and right columns are the best MAX-CUT and community detection results, respectively.}
	\label{f:imageseg}
\end{figure}

%\subsection{Sensor node localization}
%
%\[
%\begin{array}{ll}
%\maximize{X} & \tr(X)\\
%\st & \rank(X) = r\\
%& P_{\Omega}(X) = \hat X
%\end{array}
%\]
%
%
%
%\subsection{Robust PCA}
%Use Yahoo! Stock data to build a partial covariance matrix, then find a low rank completionn by solving
%\[
%\begin{array}{ll}
%\minimize{X} & \|P_{\Omega}(X-\hat X)\|_2^2\\
%\st & \rank(X) = r
%\end{array}
%\]

\paragraph{Symmetric factorization with partial observations}
Recall the factorization with partial observations formulation as follows
\begin{equation}
%\begin{array}{ll}
\minimize{Z\in \symm^n,X\in \R^{n\times r}} \displaystyle\sum_{i,j\in \Omega} (Z_{ij}-C_{ij})^2,\st  Z = XX^T,\quad X\geq 0.
%\end{array}
\label{eq:nonneg-SDP2}
\end{equation}
Note that here we generalize the aforementioned formulation with $r = 5$. In this setting, while the strongly convex $Y-$update in the proposed algorithm can no longer be solved in closed form, projected gradient descent is applied to deal with it. The relative error defined as $\|(Z^* - C)_{\Omega}\|/\|C_{\Omega}\|$ and CPU time with varying problem size and sparsity are demonstrated in Table \ref{t:nmf}.
% and Figure \ref{f:runtime-nonnegfact}.

%\red{need to compare either with GD on quartic or just use $r > 1$.}

%\begin{figure}
%\begin{center}
%\includegraphics[scale=0.75]{figs/runtime_nonnegfact.png}
%\end{center}
%\caption{Runtime of several nonnegative factorization experiments.}
%\label{f:runtime-nonnegfact}
%\end{figure}

\begin{table}
	\begin{center}
	\begingroup
	
\setlength{\tabcolsep}{2pt} % Default value: 6pt
\renewcommand{\arraystretch}{1.1} % Default value: 1
		{\footnotesize
		%\begin{tabular}{ c|cccc cccc cc  }
        \begin{tabular}{ c|ccc|ccc|ccc|ccc   }
        \hline\hline
        $n$&\multicolumn{3}{|c|}{1,000}&\multicolumn{3}{|c|}{3,000}&\multicolumn{3}{|c|}{5,000} &\multicolumn{3}{|c}{8,000}\\
        \hline
        ${|\Omega|}/{n^2}$&0.1&0.5&0.8&0.1&0.5&0.8&0.1&0.5&0.8&0.1&0.5&0.8\\
        \hline
        CPU time/s&9.74&13.53&13.97&61.15&78.99&64.76&117.54&85.24&131.64&212.26&220.42&337.74\\
        \hline
        $\frac{\|(Z^*-C)_{\Omega}\|}{\|C_{\Omega}\|}$&0.86&0.85&0.86&0.89&0.89&0.89&0.89&0.88&0.87&0.88&0.90&0.89\\\hline
        STD &0.043&0.020&0.021&0.010&0.006&0.008&0.008&0.012&0.018&0.004&0.008&0.008\\
        \hline\hline
		\end{tabular}}	
		\endgroup			
		\caption{Result for nonnegative factorization with partial observations from linearized ADMM (5 trials). STD = standard deviation.}		
		\label{t:nmf}
	\end{center}
\end{table}
%\begin{figure}[!h]
%	\vspace{-5cm}
%	\hspace{1.5cm}
%	\includegraphics[scale=0.65]{relative_obj.pdf}
%	\vspace{-6cm}
%	\caption{Relative objective versus CPU time.}
%	\label{f:relative_obj}
%\end{figure}

\section{Conclusion}

We present two methods for solving  quadratic combinatorial problems using ADMM on two reformulations. Though the problem has a nonconvex constraint, we give convergence results to KKT solutions under mild conditions. From this, we give empirical solutions to  several graph-based combinatorial problems, specifically MAX-CUT and community detection; both can can be used in additional downstream applications, like image segmentation.

\appendix

\section{Derivation of $X$, $Z$ update}

In linearized case, consider $G = \nabla f(Z^{k-1}) = G_{\Omega}$.
Then the optimality conditions of xxx are
\begin{eqnarray*}
G - \mA^*(\nu) + S + \rho(Z-(XY^T)_\Omega) &=& 0\\
U -S Y + \rho((XY^T)_\Omega Y-ZY) + \rho(X-Y) &=& 0\\
\mA(Z) &=& b.
\end{eqnarray*}
Using
$
%\begin{eqnarray*}
D =  \rho^{-1}(S Y -U) + Y , B = -\rho^{-1}(G-\mA^*(\nu) + S),
%\end{eqnarray*}
$
we get
\begin{eqnarray*}
-B + Z-(XY^T)_\Omega &=& 0\\
-D + (XY^T)_\Omega Y-ZY + X &=& 0\\
\mA(Z) &=& b.
\end{eqnarray*}
Substitute for $Z$:
$%\begin{eqnarray*}
Z = (XY^T)_\Omega + B
\Rightarrow  D +((XY^T)_\Omega + B)Y  = (XY^T)_\Omega Y+X
\Rightarrow  D + BY  = X.
$%\end{eqnarray*}
 Since we assume the diagonal is in $\Omega$, $\mA(X_\Omega) = \mA(X)$, so to solve for $\nu$:
\[
\mA((XY^T)_\Omega + B) = \mA(XY^T + B) =\mA((D + BY )Y^T + B) = b
\]
and therefore
$
\mA(B(YY^T+I)) = b - \mA(DY^T).
$
Insert $B$  and simplify
\begin{eqnarray*}
b - \mA(DY^T)  &=&  \mA((-\rho^{-1}(G-\mA^*(\nu) + S))(YY^T+I))  \\
&=&   -\rho^{-1} \mA((G-\mA^*(\nu) + S)(YY^T+I))
\end{eqnarray*}
and thus
\[
b - \mA(DY^T)  + \rho^{-1} \mA((G+ S)(YY^T+I))=    \rho^{-1} \mA(\mA^*(\nu) (YY^T+I))   = \rho^{-1}H\nu
\]
where $H$ is an $m\times m$ matrix with
$
H_{ij} = \langle A_i, A_j(YY^T + I)\rangle.
$
Thus this system reduces to
$
\nu = H^{-1} \left(b - \mA(DY^T)  + \rho^{-1} \mA((G+ S)(YY^T+I))\right).
$

\paragraph{Implicit inverse of $H$}
When $\mA = \diag$, \eqref{eq:solvenu} reduces to $n$ scalar element-wise computations
$
\nu_i = \frac{ \rho(b-(DY^T)_{ii}) + ((G+S)(I + YY^T))_{ii}}{(YY^T)_{ii}+1}.
$
When $\mA = \tr$,\\
$
\nu = \frac{\rho(b-\tr(DY^T)+\tr((G+S)(I + YY^T))}{\tr(YY^T) + 1}.
$
Note that in both cases, the computation for $\nu$ can be done without ever forming an $n\times n$ matrix. For example, for $\mA = \diag$,
$
DY^T_{ii} =
\rho^{-1}(S YY^T)_{ii} -\rho^{-1}(UY^T)_{ii} +( YY^T)_{ii}
$
Recall that for any two matrices $A$,$B\in \R^{n\times r}$, $(AB^T)_{ii} = A_i^T B_i$ where $A_i$, $B_i$ are the $i$th rows of $A$ and $B$; thus an efficient way of computing $\nu$ is
%\begin{enumerate}
i) Compute more skinny matrices $F_1 = SY$, $F_2 = GY$

ii) Compute the element-wise products $G_1 = F_1\circ Y$, $G_2 = U\circ Y$, $G_3 = F_2\circ Y$, and $G_4 = Y\circ Y$, where $(A\circ B)_{ij} = A_{ij} B_{ij}$ (element-wise multiplication).
iii) Compute the row sums $g_i = G_i\mb 1$, $i = 1,...,4$.
iv) Compute the ``numerator vector" $h_1 = \rho(b-(\rho^{-1}(g_1 -g_2) +g_4) + \diag(G)+\diag(S) + g_3+g_1$ and ``denominator vector" $h_2 = g_4+1$.
v) Then
$
\nu_i = \frac{(h_1)_i}{(h_2)_{i}}.
$

%\end{enumerate}
A similar procedure can be done for $\mA = \tr$, to keep memory requirements low.

\section{Convergence analysis for matrix form}\label{sec:matrixproof}

To simplify notation, we first collect the primal and dual variables \\ $P^k = (Z,X,Y)^k$ and $D^k = (\Lambda_1,\Lambda_2)^k$.
We define the augmented Lagrangian at iteration $k$ as
\begin{eqnarray}\label{eq:AugLag}
\mL^k:&=&\mL(P^k;D^k;\rho^k) = f(Z^{k})  + \delta_{\mC}(Y)\\
&+& \langle U,X-Y\rangle  +\langle S,Z - XY^T\rangle + \frac{\rho}{2}\|X-Y\|_F^2 + \frac{\rho}{2}\|Z - XY^T\|_F^2\nonumber
\end{eqnarray}
and its linearization at iteration $k$ as
\begin{eqnarray}\label{eq:AugLagLin}
\bar \mL^k &:=& \bar \mL(P^k;D^k;\rho^k; \bar f^k) = \bar f^k + \delta_{\mC}(Y) \\
&+&  \langle U,X-Y\rangle
 +\langle S,Z - XY^T\rangle + \frac{\rho}{2}\|X-Y\|_F^2 + \frac{\rho}{2}\|Z - XY^T\|_F^2  \nonumber
\end{eqnarray}
Here, $\bar f^k :=   f(Z^{k-1}) + \langle G^{k-1}, Z-Z^{k-1} \rangle$ such that $f^k$ is the linearization of $f$ at $Z^{k-1}$.

\begin{lemma}\label{lem:hessianboundY}
	$
	\nabla^2 \mL_Y = \nabla^2 \bar \mL_Y \succeq \rho^k I.
	$
\end{lemma}

\begin{proof}
	
		Given the definition of $\mL$, we can see that the Hessian \\
		$
		\nabla^2 \mL_Y = \rho^k\left(M+   I\right) \succeq \rho^k I
		$
		where
		$
		M = \blkdiag(X^TX,X^TX,...)\succeq 0.
		$
\end{proof}

\begin{lemma}\label{lem:hessianboundXZ}
		$
	\nabla^2 \bar\mL_{(X,Z)} \succeq \rho^k\left(1- \frac{\sqrt{\lambda_N^2+4\lambda_N}-\lambda_N}{2}\right)I.
	$
\end{lemma}

\begin{proof}
	
		For $(X,Z)$, we have
		$\nabla^2_{(X,Z)}{\mL}_k = \rho^k\begin{bmatrix}
		I + NN^T & -N \\
		-N^T & I
		\end{bmatrix}
		$
		where \\
		$
		N = \blkdiag(Y^T,\ldots,Y^T )\in\reals^{nr\times n^2}.
		$
		Note that for block diagonal matrices, $\|N\|_2 = \|Y\|_2$.
Note also that the determinant of  $\frac{1}{\rho^k} \nabla^2_{(X,Z)}{\mL}_k $ is $ \textbf{det}((I + NN^T) - NN^T) = 1 \geq \zeros$, so  $\nabla^2_{(X,Z)}\tilde{\mL}_k \succ \zeros$ and equivalently $\lambda_{\min}(\nabla^2_{(X,Z)}{\mL}_k) > 0$.
		
		To find the smallest eigenvalue $\lambda_{\min}(\nabla^2_{(X,Z)}{\mL}_k)$, it suffices to find the largest $\sigma > 0$ such that
		\begin{multline}
		H_2 = (\rho^k)^{-1}\nabla^2_{(X,Z)}\tilde{\mL}_k - \sigma I
		= \begin{bmatrix}
		(1-\sigma)I + NN^T & -N \\
		-N^T & (1-\sigma)I
		\end{bmatrix}
		\succeq\zeros.
		\end{multline}
		Equivalently,  we want to find the largest $\sigma > 0$ where $		(1-\sigma)I \succeq 0$ and the Schur complement of $H_2$
		i.4., $H_3 = (1-\sigma)I + NN^T(1-(1-\sigma)^{-1})) \succeq 0.$
		Defining $\sigma_Y = \|Y\|_2$ the largest singular vector of $Y^{k+1}$, and noting that $\lambda_{\min}(\alpha I + A) = \alpha + \lambda_{\min}(A)$ for any positive semidefinite matrix $A$, we have
		$\lambda_{\min}(H_3) = (1-\sigma) + (\sigma_Y)^2(1-(1-\sigma)^{-1}).$
		We can see that $(1-\sigma)\lambda_{\min}(H_3)$ is a convex function in $(1-\sigma)$, with two zeros at
		$1-\sigma =  \frac{\pm\sqrt{\sigma_Y^4+4\sigma_Y^2}-(\sigma_Y)^2}{2}.$
		In between the two roots, $\lambda_{\min}(H_3) <0$. Since the smaller root cannot satisfy $1-\sigma > 0$, we choose
	  $\sigma_{\max} = 1- \frac{\sqrt{\sigma_Y^4+4\sigma_Y^2}-(\sigma_Y)^2}{2} > 0$
	 	as the largest feasible $\sigma$ that maintains $\lambda_{\min}(H_3)\geq 0$. As a result,
		$%\begin{eqnarray*}
			\lambda_{\min}(\nabla^2_{(X,Z)}\bar\mL) = \rho^k\sigma_{\max}
			=\rho^k\left(1- \frac{\sqrt{\sigma_Y^4+4\sigma_Y^2}-\sigma_Y^2}{2}\right).
		$%\end{eqnarray*}
		Figure \ref{fig:xzlambdamin} shows how this term behaves according to the spectral norm of $Y$.
		
		\begin{figure}
		\begin{center}
		\includegraphics[width=3in]{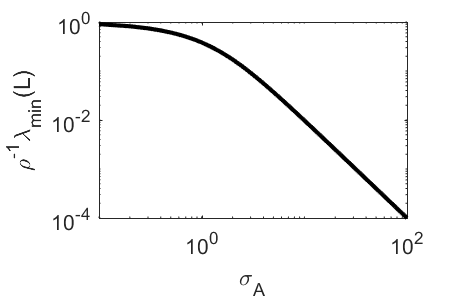}
		\end{center}
		\label{fig:xzlambdamin}
		\caption{\textbf{Strong convexity wrt $X$, $Z$.} Smallest eigenvalue of $\nabla^2_{X,Z}\mL$ as a function of the spectral norm of $Y$.}
		\end{figure}
\end{proof}

We now prove the main theorem.

\begin{lemma}\label{lemma:LagDecentXxy}
Consider the sequence
\begin{eqnarray*}
\mL^k &:=& \mL(P^k;D^k) =  f(Z^{k-1})+ \langle \nabla f(Z^{k-1}), Z-Z^{k-1} \rangle + \delta_{\mC}(Y)  \\
&+& \langle U,X-Y\rangle  +\langle S,Z - XY^T\rangle + \frac{\rho}{2}\|X-Y\|_F^2
  + \frac{\rho}{2}\|Z - XY^T\|_F^2
\end{eqnarray*}
If $f(Z)$ is $L_f$-Lipschitz smooth, then sequence $\mL^k$ generated from Alg. \ref{a:matrix} satisfies
\begin{eqnarray}\label{eq:deltaXYZ}
\mL^{k+1}-\mL^k&\leq&
-c_1^k\|{X}^{k+1}-X^k\|_F^2
-c_2^k\|{Z}^{k+1}-{Z}^k\|_F^2
-c_3^k\|Y^{k+1}-Y^k\|_F^2\nonumber\\
&&\quad +\frac{\rho^{k+1}+\rho^k}{2(\rho^k)^2}\big(\|{S}^{k+1}-{S}^{k}\|_F^2 +
 \|{U}^{k+1}-{U}^{k}\|_F^2 \big).
\end{eqnarray}
with
 $c_1^k = \frac{\rho^k}{2}\left(1- \frac{\sqrt{\sigma_Y^4+4\sigma_Y^2}-\sigma_Y^2}{2}\right) $, $c_2^k =c_1^k- \frac{L_f}{2} $, and $c_3^k = \frac{\rho^{k}}{2} > 0$.

\end{lemma}

\begin{proof}
The proof outline of Lemma \ref{lemma:LagDecentXxy} is to show that each update step is a non-ascent step in the linearized augmented Lagrangian, and at least one update step is descent. We can describe the linearized ADMM in terms of four groups of updates:
 the primal variable $Y$, the primal variables $X$ and $Z$, the dual variables $U$, $S$, and coefficient $\rho$.

In other words, at iteration $k$, taking
%\begin{enumerate}
i) $\mL^k =  \mathcal{L}(Z^{k},X^{k},Y^{k};D^k;\rho^k; G^k) $,
ii) $\mL^Y = \mathcal{L}(Z^{k},X^{k},Y^{k+1};D^k;\rho^k; G^k)$,
iii) $\mL^{XZ} = \mathcal{L}(P^{k+1};D^k;\rho^k; G^k)$, and
iv) $\mL^{k+1} = \mathcal{L}(P^{k+1};D^{k+1};\rho^{k+1}; G^k)$
%\end{enumerate}
and $
\mL^{k+1} - \mL^k = (\mL^Y - \mL^k) + (\mL^{XZ} - \mL^Y) + (\mL^{k+1}-\mL^{XZ}).
$
We now lower bound each term.

\begin{enumerate}
\item \underline{Update $Y$.}
For the update of $Y$ in (\ref{eq:xnewXxy}), taking \\ $\mL^Y = \mathcal{L}(Z^{k},X^{k},Y^{k+1};D^k;\rho^k; G^k) $,  we have
\begin{eqnarray}
\mL^Y - \mL^k &\overset{(a)}{\leq}& \langle \nabla_Y\mL^Y,Y^{k+1}-Y^k\rangle-\frac{\lambda_{\min}(\nabla^2_{\text{vec}Y}\mL^Y)}{2}\|Y^{k+1}-Y^k\|_F^2\nonumber\\
&\overset{(b)}{\leq}& -\frac{\rho^{k}}{2}\|Y^{k+1}-Y^k\|_F^2 \label{eq:D2Xxy}
\end{eqnarray}
where (a) follows from the definition of strong convexity, and (b)  the optimality of $Y^{k+1}$.

%with $\nabla^2_{Y}\bar{\mL} = \nabla^2_{\text{vec}(Y)}\mL(Z^{k},X^{k},Y^{K+1};D^k;\rho^{k}) \succeq \frac{\rho^{k}}{2}I$.
\item \underline{Update $X$, $Z$.}
Similarily, the update of $(Z,X)$ in (\ref{eq:xnewXxy}), denoting \\ $\mL^{XZ} = \mathcal{L}(P^{k+1}; D^k;\rho^{k}; G^k)$, we have
\begin{eqnarray}
\bar \mL^{XZ} - \mL^Y &\overset{(a)}{\leq} &\langle \nabla_{{Z}}\bar \mL^{XZ},{Z}^{k+1}-{Z}^k\rangle  + \langle \nabla_{X}\bar\mL^{XZ},X^{k+1}-X^k\rangle  \nonumber \\
&&\quad - \frac{\lambda_{\min}(\nabla^2_{(X,Z)}\mL^{XZ})}{2}\big(\|{Z}^{k+1}-{Z}^k\|_F^2 + \|X^{k+1}-X^k\|_F^2\big)\nonumber\\
&\overset{(b)}{\leq}&  -\frac{\lambda_{\min}(\nabla^2_{(X,Z)}\bar\mL^{XZ})}{2}\left(\|{Z}^{k+1}-{Z}^k\|_F^2 + \|X^{k+1}-X^k\|_F^2\right),\label{eq:D1Xxy}
\end{eqnarray}
where (a) follows from the definition of strong convexity, and (b)  the optimality of $X^{k+1}$ and $Z^{k+1}$.
To further bound $\mL^{XZ}-\bar\mL^{XZ}$, we use the linearization definitions
\begin{eqnarray}
\mL^{XZ}-\bar\mL^{XZ} &=& f(Z^{k+1})-f(Z^{k}) - \langle \nabla f(Z^{k}) , Z^{k+1} - Z^k\rangle \nonumber\\
&\overset{(a)}{=}& \leq \frac{L_f}{2}\|Z^{k+1}-Z^k\|_F\label{eq:D1XxyLin}
\end{eqnarray}
where (a) comes from the $L_f$ Lipschitz smooth property of $f$.

\item \underline{Update $S$, $U$, and $\rho$.}
For the update of the dual variables and the penalty coefficient, with $\mL^k = \mL(P^k;D^k;\rho^k)$, we have
\bea\label{eq:D3Xxy}
\mathcal{L}^{D} -  \mathcal{L}^{XZ}\overset{(a)}{=}&  \langle S^{k+1}-S^{k},Z^{k+1}-X^{k+1}(Y^{k+1})^T\rangle
 + \langle U^{k+1}-U^{k},X^{k+1}-Y^{k+1}\rangle \nonumber\\
 +&\frac{\rho^{k+1}-\rho^k}{2}(\|Z^{k+1}-X^{k+1}(Y^{k+1})^T\|_F^2)
 +\frac{\rho^{k+1}-\rho^k}{2}(\|X^{k+1} - Y^{k+1}\|_F^2)\nonumber\\
&\overset{(b)}{=}  \frac{\rho^{k+1}+\rho^k}{2(\rho^k)^2}\big(\|S^{k+1}-S^{k}\|_F^2
  + \|U^{k+1}-U^{k}\|_F^2 \big)
\eea
where the (a) follows the definition of $\mathcal{L}$ and (b) from the dual update procedure.

%Recall the definition of $\mathcal{L}$ and $\mathcal{L}_k$, and it leads to that
%\begin{align}\label{eq:D1Xxy2}
%&\mathcal{L}_{k+1}^\prime(P^{k+1}; D^k;\rho^{k})  =  \mathcal{L}(P^{k+1}; D^k;\rho^{k}) - f(Z^{k+1}) + \langle \nabla f(Z^{k}), Z-Z^{k} \rangle\nonumber\\
%&\mathcal{L}_{k}^\prime(Z^{k},X^{k},Y^{k+1};D^k;\rho^{k}) = \mathcal{L}(Z^{k},X^{k},Y^{k+1};D^k;\rho^{k}) - f(Z^{k})
%\end{align}
%Combining \eqref{eq:D1Xxy} and \eqref{eq:D1Xxy2} yields
%\begin{align}\label{eq:D1Xxy3}
%&\mathcal{L}(P^{k+1}; D^k;\rho^{k}) - \mathcal{L}(Z^{k},X^{k},Y^{k+1};D^k;\rho^{k}) + f(Z^{k}) - f(Z^{k+1}) + \langle \nabla f(Z^{k}), Z-Z^{k} \rangle \nonumber\\
%\le\;&  -\frac{\lambda_{\min}(\nabla^2_{(X,Z)}\tilde{\mL}_k)}{2}\big(\|{Z}^{k+1}-{Z}^k\|_F^2 + \|X^{k+1}-X^k\|_F^2\big),
%\end{align}
%
%%Considering Assumption \ref{Assumption:Boundedf}, \eqref{eq:D1Xxy3} further results in
%\begin{align}\label{eq:D1Xxy4}
%&\mathcal{L}(P^{k+1}; D^k;\rho^{k}) - \mathcal{L}(Z^{k},X^{k},Y^{k+1};D^k;\rho^{k}) \nonumber\\
%\le & -\frac{\lambda_{\min}(\nabla^2_{(X,Z)}\tilde{\mL}_k) - L_f}{2}\big(\|{Z}^{k+1}-{Z}^k\|_F^2\big) - \frac{\lambda_{\min}(\nabla^2_{(X,Z)}\tilde{\mL}_k)}{2}\big(\|X^{k+1}-X^k\|_F^2\big),
%\end{align}
%with $\lambda_{\min}(\nabla^2_{(X,Z)}\tilde{\mL}_k) - L_f > 0$ given $\rho^k\left(1- \frac{\sqrt{\lambda_N^2+4\lambda_N}-\lambda_N}{2}\right) > L_f$.
\end{enumerate}

The lemma statement results  by incorporating (\ref{eq:D1Xxy}), (\ref{eq:D2Xxy}), (\ref{eq:D3Xxy}), and \eqref{eq:D1XxyLin}.
\end{proof}

\begin{lemma}
If $\mL_k$ is unbounded below, then either problem \eqref{eq:main} is unbounded below, or the sequence $L_f\|Z_k-Z_{k-1}\|_F$ diverges.
\end{lemma}
\begin{proof}
First, consider the case that $\mL_k$ is unbounded below. First rewrite $\mL^k$ equivalently as
\begin{eqnarray*}
\mL^k &=& f(Z^{k-1})+ \langle \nabla f(Z^{k-1}), Z^k-Z^{k-1} \rangle + \delta_{\mC}(Y^k)  +
 \frac{\rho}{2}\|X^k-Y^k + \frac{1}{\rho^k}U^k\|_F^2
  \\&&\quad+ \frac{\rho}{2}\|Z^k - X^k(Y^k)^T + \frac{1}{\rho^k}S^k\|_F^2
  -\frac{1}{2\rho^k}\|U^k\|_F^2-\frac{1}{2\rho^k}\|S^k\|_F^2.
\end{eqnarray*}
Since $\|U^k\|_F$ and $\|S^k\|_F$  are bounded above, this implies that the linearization
$g^k := f(Z^{k-1})+ \langle \nabla f(Z^{k-1}), Z^k-Z^{k-1} \rangle$ is unbounded below.

Note that
\[
g^k - f(Z^k) = f(Z^{k-1})-f(Z^k) - \nabla f(Z^{k-1}),Z^{k-1}-Z^k\rangle \geq - \frac{L_f}{2}\|Z^k-Z^{k-1}\|_F^2
\]
which implies either $f(Z^k)\to-\infty$ or $L_f\|Z^k-Z^{k-1}\|_F^2\to+\infty$.
\end{proof}
\begin{corollary}
If $\mL_k$ is unbounded below and the objective $f(Z) = \tr(CZ)$ then it must be that \eqref{eq:main} is unbounded below. This follows immediately since $L_f = 0$.
\end{corollary}

\begin{theorem}
Assume the dual variables are bounded,
 e.g. \\
$
\max\{\|S^k\|_F,\|U^k\|_F,\|Y^k\|_F\}_k \leq B_P < +\infty,
$
and
$\frac{L_f}{\sigma_{\max}}$ is bounded above, where
$
\sigma_{\max} = 1 - \frac{\sqrt{\sigma_Y^4 + 4\sigma_Y^2} - \sigma_Y^2}{2},\quad \sigma_Y = \|Y^{k+1}\|_2.
$
Then by running  Alg. \ref{a:matrix} with $\rho^k = \alpha \rho^{k-1}=\alpha^k\rho_0$, if $\mL_k$ is bounded below, then
the sequence  $\{P^k,D^k\}$ converges to a stationary point of \eqref{eq:AugLag-1}.

\end{theorem}

\begin{proof}

If $f(Z)$ is linear, take $K_0 = 0$. If $f(Z)$ is $L_f>$ smooth, take $\hat K$ large enough such that for all $k > K_0$, $\alpha^k \rho \geq L_f \sigma_{\max}$. By assumption, $K_0$ is always finite.

Taking $\Delta_{XYZ}^k = \left(\|{Z}^{k+1}-{Z}^k\|_F^2+\|X^{k+1}-X^k\|_F^2
+\|Y^{k+1}-Y^k\|_F^2 \right)$ \\ and $c^k = \min\{c_1,c_2,c_3\}$ , the summation of \eqref{eq:deltaXYZ} leads to
\begin{eqnarray}\label{eq:D6Xxy}
\mL^K- \mL^{K_0} &=&
\sum_{k=K_0}^{K-1}\mL^{k+1} - \mL^k \nonumber\\
&\leq&\sum_{k=K_0}^{K-1}\frac{\rho^{k+1}+\rho^k}{2(\rho^k)^2}\left(\|S^{k+1} - S^k\|_F^2 + \|U^{k+1} - U^k\|_F^2\right)
-\sum_{k=K_0}^{K-1}c^k\Delta_{XYZ}^k\nonumber\\
&\overset{(a)}{\leq}&
4B_P\sum_{k=K_0}^{K-1}\frac{\rho^{k+1}+\rho^k}{2(\rho^k)^2}
-\sum_{k=K_0}^{K-1}c^k\Delta_{XYZ}^k 
\overset{(b)}{\leq} 4B_P\sum_{k=K_0}^{K-1}\frac{\rho^{k+1}+\rho^k}{2(\rho^k)^2}\nonumber
\end{eqnarray}
where (a) follows from  the boundedness assumption of the dual variables, and
 and (b)
follows from Lemma \ref{lem:hessianboundXZ}, \ref{lem:hessianboundY}, and careful construction of $\rho$ with respect to $L_f$ and $\|Y^{k+1}\|_2$.
Further simplifying, we see that $L^K$ is thus bounded above, since
\begin{eqnarray*}
\mL^K- \mL^{K_0}&\leq&  \lim_{K\to\infty}4B_P\sum_{k={K_0}}^{K-1}\frac{\rho^{k+1}+\rho^k}{2(\rho^k)^2} \\
&=& 4B_P\frac{1+\alpha}{2\alpha^{K_0}\rho}\left(1+\frac{1}{\alpha}+\frac{1}{\alpha^2}+\cdots\right)
= \frac{4B_P}{2\alpha^{K_0}\rho} < +\infty.
\end{eqnarray*}

If $\mL^k$ is not unbounded below, then
\begin{equation}
0\leq \sum_{k=K_0}^{K-1}(c_1 \|X^{k+1}-X^k\|_F^2 + c_2 \|{Z}^{k+1}-{Z}^k\|_F^2
+c_3 \|Y^{k+1}-Y^k\|_F^2 \bigg)\leq +\infty.
\end{equation}
Recall $c_3^k = \frac{\rho^k}{2}$, and by boundedness assumption on $\|Y^{k+1}_2$, for $k > K_0$, $c_1^k,c_2^k \propto \rho^k$.
Since additionally $\sum_k\rho_k = +\infty$, then this immediately yields
${Z}^{k+1}-{Z}^k\to 0, X^{k+1}-X^k\to 0, Y^{k+1}-Y^k\to 0$.

Therefore, since the primal variables are convergent, this implies that
\[
 Z^{k+1}-(X^{k+1}(Y^{k+1})^T)_\Omega = \frac{1}{\rho^k}(S^{k+1} - S^k),\quad
X^{k+1}-Y^{k+1} = \frac{1}{\rho^k}(U^{k+1} - U^k)
\]
converges to a constant. But since $\rho^k\to \infty$ and the dual variables are all bounded,
 then it must be that
$
 Z^{k+1}-(X^{k+1}(Y^{k+1})^T)_\Omega \to 0,
X^{k+1}-Y^{k+1} \to 0.
 $
 Therefore the limit points $X^*, Y^*$, and $Z^*$ are all feasible, and simply checking the first optimality condition will verify that this accumulation point is a stationary point of (\ref{eq:AugLag-1}).
\end{proof}

\section{Convergence analysis for vector form}\label{sec:vectorconvergence}
%%%%%%%%%%%%%%%%%

\begin{lemma}\label{lemma:Lambda}
	For  two adjacent iterations of Algorithm \ref{a:vector} we have
	 \bea\label{eq:OptCond3}
	\|U^{k+1} - U^{k}\|_2^2 \le L_g^2 \|X^{k+1} - X^{k}\|_2^2.
	\eea
\end{lemma}
\begin{proof}
	From the first order optimality conditions for the update of $X$
	\bea\label{eq:OptCond0}
	\nabla g(X^{k+1}) + U^k + \rho^k(X^{k+1}-Y^{k+1}) = 0.
	\eea
	Combining with the dual update, we get
$
	\nabla g(X^{k+1}) + U^{k+1} = 0.
$ Then result follows from the definition of $L_g$.
\end{proof}

Next we will show that the augmented Lagrangian  is monotonically decreasing and lower bounded.
\begin{lemma}\label{lemma:auglag}
	Each step in the augmented Lagrangian update is decreasing, e.g. for
	\begin{equation}
	\mL(X,Y;U;\rho) := g(X) + \delta_{\mC}(Y) + \langle U, X-Y\rangle + \frac{\rho}{2}\|X-Y\|_F^2
	\label{eq:auglag-simple}
	\end{equation}
	we have
	\begin{eqnarray}
	\mathcal{L}(Y^{k+1},X^{k+1};U^{k+1};\rho^{k+1}) &\leq&
	\mathcal{L}(Y^{k+1},X^{k+1};U^{k};\rho^{k}) \nonumber \\
	 &\leq&	\mathcal{L}(Y^{k+1},X^{k};U^{k};\rho^{k})  \leq
	\mathcal{L}(Y^{k},X^{k};U^{k};\rho^{k}).\label{eq:simple-auglagdecrease}
	\end{eqnarray}
	Furthermore, the amount of decrease is
	\begin{eqnarray}
	\mathcal{L}(Y^{k+1},X^{k+1};U^{k+1};\rho^{k+1}) -  \mathcal{L}(Y^{k},X^{k};U^{k};\rho^{k})  \nonumber \\
	\leq  - \rho^k\|Y^{k+1}-Y^k\|_F^2
	-c^k\|X^{k+1}-X^k\|_F^2.
	\label{eq:simple-descent-condition}
		\end{eqnarray}
		Here,
		\begin{itemize}
		\item
	 if $g(X)$ is $H_g$-strongly convex (where $H_g = 0$ if $g$ is convex but not strongly convex)
	 then $c^k = \frac{\rho^k+H_g}{2} - L_g^2\frac{\rho^{k+1}+\rho^{k}}{2(\rho^k)^2}$, and
	 \item if $g(X)$ is nonconvex but $L_g$-smooth, then
	     $c^k = \frac{\rho^k-3L_g}{2} - L_g^2\frac{\rho^{k+1}+\rho^{k}}{2(\rho^k)^2}$.
	     \end{itemize}
\end{lemma}

\begin{proof}
Both the updates of $Y$ and $X$ globally minimize $\mL$ with respect to those variables. To minimize $Y$ at $(X,U) = (X^k,U^k)$:
\begin{eqnarray}
		\mathcal{L}(Y^{k+1},{X};{U};\rho) - \mathcal{L}(Y^k,{X};{U};\rho)
		&\overset{(a)}{\leq}& \langle \nabla_{Y}\mL(Y^{k+1},{X};{U}; \rho),{Y}^{k+1}-{Y}^k\rangle
				-\frac{\rho^k}{2}\|Y^{k+1}-Y^k\|_2^2 \nonumber\\
		&\overset{(b)}{\leq}& -\frac{\rho^k}{2}\|{Y}^{k+1}-{Y}^k\|_2^2.
	\end{eqnarray}
	
To minimize $X$ at $(Y,U) = (Y^{k+1},U^k)$, we consider two cases. If $g$ is $H_g$-strongly convex, then
\begin{eqnarray}
		\mathcal{L}(Y,{X}^{k+1};{U};\rho) - \mathcal{L}(Y,{X}^k;{U};\rho)
		&\overset{(a)}{\leq}& \langle \nabla_{{X}}\mL(Y,{X}^{k+1};{U}; \rho),{X}^{k+1}-{X}^k\rangle
		\nonumber\\&&\quad
		-\frac{\rho^k+H_g}{2}\|{X}^{k+1}-{X}^k\|_2^2   \nonumber\\
		&\overset{(b)}{\leq}& -\frac{\rho^k+H_g}{2}\|{X}^{k+1}-{X}^k\|_2^2
	\end{eqnarray}
	where (a) follows from the strong convexity of  $\mL(Y,{X};U;\rho)$ with respect to $X$, and (b) follows from the optimality condition of the update.
	If $g$ is nonconvex but $L_g$-Lipschitz, then note that	
\begin{eqnarray*}
g(X^{k+1}) - g(X^k) &\leq& \langle \nabla g(X^k),X^{k+1}-X^k\rangle + \frac{L_g}{2}\|X^{k+1}-X^k\|_F^2\\
&\overset{(a)}{=}& \langle \nabla g(X^k)- \nabla g(X^{k+1}),X^{k+1}-X^k\rangle
\nonumber\\&&\quad
+ \frac{L_g}{2}\|X^{k+1}-X^k\|_F^2
+ \langle\nabla g(X^{k+1}),X^{k+1}-X^k\rangle\\
&\overset{(b)}{\leq}& \| \nabla g(X^k)- \nabla g(X^{k+1})\|_F\|X^{k+1}-X^k\|_F
\nonumber\\&&\quad
+ \frac{L_g}{2}\|X^{k+1}-X^k\|_F^2
+ \langle\nabla g(X^{k+1}),X^{k+1}-X^k\rangle\\
&\overset{(c)}{\leq}& \frac{3L_g}{2}\|X^{k+1}-X^k\|_F^2
+\langle \nabla g(X^{k+1}),X^{k+1}-X^k\rangle
\end{eqnarray*}
where (a) follows from adding and subtracting a term, (b) from Cauchy-Schwartz, and (c) from the Lipschitz gradient condition on $g$.
Therefore
\begin{eqnarray*}
		\mathcal{L}(Y,{X}^{k+1};{U};\rho) - \mathcal{L}(Y,{X}^k;{U};\rho)
		&\overset{(a)}{\leq}& \langle \nabla_{{X}}\mL(Y,{X}^{k+1};{U}; \rho),{X}^{k+1}-{X}^k\rangle
		\nonumber\\&&\quad
		-\frac{\rho^k-3L_g}{2}\|{X}^{k+1}-{X}^k\|_2^2   \nonumber\\
		&\overset{(b)}{\leq}& -\frac{\rho^k-3L_g}{2}\|{X}^{k+1}-{X}^k\|_2^2.
	\end{eqnarray*}

	In the dual variables,  using $\{X,Y\} = \{X^{k+1},Y^{k+1}\}$ we have
	\begin{eqnarray*}
	\mathcal{L}(Y,{X};{U}^{k+1};\rho^{k+1}) - \mathcal{L}(Y,{X};{\mu}^{k};\rho^{k})
	&\overset{(a)}{\leq}&\langle {U}^{k+1}-{U}^{k},X - Y\rangle \nonumber + \frac{\rho^{k+1}-\rho^{k}}{2}\|{X} - Y\|_F^2\\
	&\overset{(b)}{\leq}& \frac{\rho^{k+1}+\rho^{k}}{2(\rho^k)^2}\|{U}^{k+1}-{U}^{k}\|_2^2 \nonumber\\
	&\overset{(c)}{\leq} &  L_g^2\;\frac{\rho^{k+1}+\rho^{k}}{2(\rho^k)^2}\|{X}^{k+1}-{X}^{k}\|_2^2
	\end{eqnarray*}
	where (a) follows the definition of $\mathcal{L}$,  (b) follows from the update of $U$, and (c) follows from  Lemma \eqref{lemma:Lambda} since $\rho^k > 0$ for al $k$.
	Incorporating these observations completes the proof.
	\end{proof}

	\begin{lemma}
	    If $\rho^k \geq L_g$ and the objective $g(X)$ is lower-bounded over $\mC$, then the augmented Lagrangian \eqref{eq:auglag-simple} is lower bounded.
	\label{lem:simpledescentlemma}
	\end{lemma}

	\begin{proof}
	%In the following, we will prove that $\mathcal{L}({z}^{k},{x}^k;{U}^k)$ is lower bounded and thus convergent.
 From the $L_g$-Lipschitz continuity of  $\nabla g(X)$ , it follows that
	\bea\label{eq:LHUpp}
	g(X) \geq g(Y) + \langle \nabla g(X),X - Y \rangle -\frac{L_g}{2}\|X - Y\|_F^2
	\eea
	for any $X$ and $Y$. By definition
	\begin{eqnarray}
	\label{eq:CRCP9}
	\mL(Y^{k},{X}^k;{U}^{k};\rho^{k}) &=& g({X}^{k}) + \langle {U}^{k},{X}^{k} - Y^{k}\rangle + \frac{\rho^{k}}{2}\|{X}^{k}-Y^{k}\|_F^2\nonumber\\
	&\overset{(a)}{=}&  g({X}^{k}) - \langle \nabla g(X^k),{X}^{k} - Y^{k}\rangle + \frac{\rho^{k}}{2}\|{X}^{k}-Y^{k}\|_F^2 \nonumber\\
	&\overset{(b)}{\geq}&  g(Y^{k}) + \frac{\rho^{k}-L_g}{2}\|{X}^{k}-Y^{k}\|_F^2 ,
	\end{eqnarray}
	where (a) follows from the optimality in updating $X$ and (b) follows from (\ref{eq:LHUpp}).
	Since $\mL^k$ is unbounded below, then $g(Y^k)$ is unbounded below. Since $Y^k\in \mC$ for all $k$, this implies that $g$ is unbounded below over $\mC$.
\end{proof}

Thus, if $g(X)$ is lower-bounded over $\mC$, then since the sequence $\{\mL({X}^{k},{Y}^k;{U}^k)\}$ is monotonically decreasing and lower bounded, then the sequence $\{\mL({X}^{k},{Y}^k;{U}^k)\}$ converges. Given the monotonic descent of each subproblem (Lemma \ref{lemma:auglag}) and strong convexity of $\mL^k$ with respect to $X$ and $Y$, it is clear that $X^k\to X^*$, $Y^k\to Y^*$ fixed points. Combining with Lemma \ref{lemma:Lambda} gives also $U^k\to U^*$.
%	 we have
%	\bea
%	\|{x}^{k+1} - {x}^{k}\| \to 0,
%	\eea
%	which indicates that the sequence $\{{x}^k\}$ converges to a limit point ($x^*$). Moreover, as $\{{x}^k\}$ converges, the sequence $\{{U}^{k}\}$ also converges to $U^*$ based on (\ref{eq:OptCond3}).
%	From (\ref{eq:LambdaUpd}), ${U}^{k+1} - {U}^{k} = \rho^k ({x}^{k+1} - Y^{k+1})$ and so  ${x}^{k+1} - Y^{k+1} \to 0$; thus $\{Y^{k}\}$ converges to a limit point $Y^*$ and ${x}^*=Y^*$. Additionally, note that $x^* = y^*$ implies
	%In conclusion, $\lim_{k\to\infty}\{Y^k,{x}^k,{U}^k\} = \{Y^*,{x}^*,{U}^*\}$ and ${x}^* = Y^*$.
%	 $\nabla_x \mL(y^*,x^*;U^*;\rho) = \nabla_y \mL(y^*,x^*;U^*;\rho) = 0$ and $x^* = y^* \in \{-1,1\}$ are feasible.
%Thus  Theorem \ref{lemma:LagDecent} is proven.

The proof of Theorem \ref{thm:simple-convergence} easily follows from Lemma \ref{lem:simpledescentlemma}.

\subsection{Linear rate of convergence when $g$ is strongly convex}\label{sec:vectorconvergence-linear}

\begin{lemma}
Consider Alg. \ref{a:vector} with $\rho^k$ constant. Then
collecting the variables all vectorized $x = (X,Y,Y)$,
\[
\mL^{k+1}-\mL^k \leq -c_3\|x^{k+1}-x^k\|^2,
\]
where  $g$ is $H_g$ strongly convex and
\[
c_3 = \max_{\theta\in(0,1)}\min\left\{\theta \left(\frac{\rho+H_g}{2} - \frac{L_g^2}{\rho}\right), (1-\theta )\left(\frac{\rho+H_g}{2H_g} - \frac{L_g^2}{\rho H_g}\right),
-\rho\right\}.
\]

\end{lemma}

\begin{proof}
From Lemma \ref{lemma:auglag} we already have that
\[
\mL^{k+1}-\mL^k \leq -\rho\|Y^{k+1}-Y^{k}\|^2-c\|X^{k+1}-X^{k}\|^2
\]
where for constant $\rho$,
$c =  \frac{\rho+H_g}{2} - \frac{L_g^2}{\rho}$.
Moreover, when $g(X)$ is $H_g$-strongly convex,
\[
\|U^{k+1}-U^k\|_2 = \|\nabla g(X^{k+1})-\nabla g(X^k)\|_2 \geq H_g\|X^{k+1}-X^k\|_2.
\]
Therefore
\[
\mL^{k+1}-\mL^k \leq -\theta\frac{c}{H_g}\|U^{k+1}-U^{k}\|_2^2-(1-\theta)c\|X^{k+1}-X^{k}\|^2.
\]
for any $\theta \in (0,1)$,
We thus have
\begin{eqnarray*}
\mL^{k+1}-\mL^k &\leq & -\theta c\|X^{k+1}-X^{k}\|_F^2
- (1-\theta )\frac{c}{H_g}\|U^{k+1}-U^{k}\|_F^2
-\rho\|Y^{k+1}-Y^{k}\|_F^2\\
&\leq & -\min\left\{\theta c, (1-\theta )\frac{c}{H_g},
-\rho\right\}
\bmat X^{k+1}-X^k\\Y^{k+1}-Y^k\\U^{k+1}-U^k\emat.
\end{eqnarray*}
\end{proof}

Note that this does not mean $\mL$ is strong convex with respect to the \emph{collected} variables $x = (X,Y,Z)$ ($\mL$ is not even convex). But with respect to each variable $X$, $Y$, and $Z$, it is strongly convex.

\begin{lemma}
Again with $\rho^k>1$ constant and collecting $x = (X,Y,Z)$, we have
\[
\mL^{k+1}-\mL^* \leq c_4\|x^{k+1}-x^*\|^2,
\quad c_4 = \min\{L_g+\rho+2,2\rho,1\}
\]
whenever $Y^{k+1}$ and $Y^*$ are both in $\mC$.
\end{lemma}
\begin{proof}
Over the domain $\mC$, the augmented Lagrangian can be written as
\[
\mL(x) = g(X) + \langle U, X-Y\rangle + \frac{\rho}{2}\|X-Y\|_F^2
\]
with gradient
$
\nabla \mL(x) =
\bmat \nabla g(X) + U + \rho(X-Y)\\
-Y+\rho(Y-X)\\
X-Y
\emat
$
and thus
\begin{eqnarray*}
\|\nabla \mL(x_1)-\nabla \mL(x_2)\|^2_F &=&
\|\nabla_X\mL(x_1)-\nabla_X\mL(x_2)\|_F^2
\nonumber\\&&\quad +
\|\nabla_Y\mL(x_1)-\nabla_Y\mL(x_2)\|_F^2 +
\|\nabla_U\mL(x_1)-\nabla_U\mL(x_2)\|_F^2\\
&\leq &   (L_g+\rho+2)\|X_1-X_2\|_F^2 + (2\rho)\|Y_1-Y_2\|_F^2 + \|U_1-U_2\|_F^2\\
&\leq& \min\{L_g+\rho+2,2\rho,1\}\|x_2-x_1\|_2^2
\end{eqnarray*}
which reveals the Lipschitz smoothness constaint for $\mL$ as
$
c_4 = \min\{L_g+\rho+2,2\rho,1\}.
$
Then using first-order optimality conditions,
\begin{eqnarray*}
\mL^{k+1} &\leq& \mL^* + \langle \nabla \mL(x^*),x^{k+1}-x^*\rangle +  c_4\|x^{k+1}-x^*\|_2^2\\
&\overset{(a)}{\leq} & \mL^*  +  c_4\|x^{k+1}-x^*\|_2^2
\end{eqnarray*}
where (a) follows from the optimality of $\mL^*$.

\end{proof}

\begin{lemma}
Consider $g(x)$ $H_g$-strongly convex in  $x$, and $\rho$ large enough so that $c_3 > 0$. Then
 the number of steps for $|\mL^k-\mL^0|\leq \epsilon$ is  $O(\log(1/\epsilon))$.
\end{lemma}
This proof is standard in the linear convergence of block coordinate descent when the objective is strongly convex. Note that $\mL$ is not strongly convex or even convex, but still all the steps hold.
\begin{proof}
Take $x^k = \{X^k, Y^k, U^k\}$ and $x^* = \{X^*, Y^*, U^*\}$. Then
\begin{eqnarray*}
\mL(x^k)-\mL(x^*) &=&\mL(x^k)-\mL(x^{k+1}) + \mL(x^{k+1}) - \mL(x^*)\\
&\geq& c_3\|x^{k+1}-x^k\|^2 + \mL(x^{k+1}) - \mL(x^*)\\
&\geq& \left(\frac{c_3}{c_4}+1\right)( \mL(x^{k+1}) - \mL(x^*))
\end{eqnarray*}
Therefore
\[
\frac{\mL(x^k)-\mL(x^*)}{\mL(x^0)-\mL(x^*)} \leq \left(\frac{c_4}{c_4+c_3}\right)^k
\]
and so
\[
\mL(x^k)-\mL(x^*)\leq \epsilon
\]
if
\[
k \geq D_1\log(1/\epsilon) + D_2
 \]
 where
 \[
 D_1 = \log^{-1}\left(\frac{c_4+c_3}{c_4}\right),\qquad D_2 = \frac{\log(\mL(x^0)-\mL(x^*))}{\log\left(\frac{c_4+c_3}{c_4}\right)}.
 \]

\end{proof}

%\section*{Acknowledgments} %we add this in after we get accepted, I think, to maintain anonymity.
%\red{grant acknowledgements go here}
%This work was done while at Technicolor Research, Palo Alto, CA.

\bibliographystyle{siamplain}
\bibliography{refs}

\begin{thebibliography}{10}

\bibitem{abbe2016exact}
{\sc E.~Abbe, A.~S. Bandeira, and G.~Hall}, {\em Exact recovery in the
  stochastic block model}, IEEE Transactions on Information Theory, 62 (2016),
  pp.~471--487.

\bibitem{anjos2012introduction}
{\sc M.~F. Anjos and J.~B. Lasserre}, {\em Introduction to semidefinite, conic
  and polynomial optimization}, in Handbook on semidefinite, conic and
  polynomial optimization, Springer, 2012, pp.~1--22.

\bibitem{bandeira2016low}
{\sc A.~S. Bandeira, N.~Boumal, and V.~Voroninski}, {\em On the low-rank
  approach for semidefinite programs arising in synchronization and community
  detection}, in Conference on Learning Theory, 2016, pp.~361--382.

\bibitem{bao2011semidefinite}
{\sc X.~Bao, N.~V. Sahinidis, and M.~Tawarmalani}, {\em Semidefinite
  relaxations for quadratically constrained quadratic programming: A review and
  comparisons}, Mathematical programming, 129 (2011), pp.~129--157.

\bibitem{barahona1988application}
{\sc F.~Barahona, M.~Gr{\"o}tschel, M.~J{\"u}nger, and G.~Reinelt}, {\em An
  application of combinatorial optimization to statistical physics and circuit
  layout design}, Operations Research, 36 (1988), pp.~493--513.

\bibitem{barvinok1995problems}
{\sc A.~I. Barvinok}, {\em Problems of distance geometry and convex properties
  of quadratic maps}, Discrete \& Computational Geometry, 13 (1995),
  pp.~189--202.

\bibitem{blekherman2012semidefinite}
{\sc G.~Blekherman, P.~A. Parrilo, and R.~R. Thomas}, {\em Semidefinite
  optimization and convex algebraic geometry}, SIAM, 2012.

\bibitem{boumal2016non}
{\sc N.~Boumal, V.~Voroninski, and A.~Bandeira}, {\em The non-convex
  {B}urer-{M}onteiro approach works on smooth semidefinite programs}, in
  Advances in Neural Information Processing Systems, 2016, pp.~2757--2765.

\bibitem{boyd2011distributed}
{\sc S.~Boyd, N.~Parikh, E.~Chu, B.~Peleato, and J.~Eckstein}, {\em Distributed
  optimization and statistical learning via the alternating direction method of
  multipliers}, Foundations and Trends{\textregistered} in Machine Learning, 3
  (2011), pp.~1--122.

\bibitem{burer2003nonlinear}
{\sc S.~Burer and R.~D. Monteiro}, {\em A nonlinear programming algorithm for
  solving semidefinite programs via low-rank factorization}, Mathematical
  Programming, 95 (2003), pp.~329--357.

\bibitem{burer2005local}
{\sc S.~Burer and R.~D. Monteiro}, {\em Local minima and convergence in
  low-rank semidefinite programming}, Mathematical Programming, 103 (2005),
  pp.~427--444.

\bibitem{burer2008finite}
{\sc S.~Burer and D.~Vandenbussche}, {\em A finite branch-and-bound algorithm
  for nonconvex quadratic programming via semidefinite relaxations},
  Mathematical Programming, 113 (2008), pp.~259--282.

\bibitem{candes2015phase}
{\sc E.~J. Candes, Y.~C. Eldar, T.~Strohmer, and V.~Voroninski}, {\em Phase
  retrieval via matrix completion}, SIAM review, 57 (2015), pp.~225--251.

\bibitem{candes2009exact}
{\sc E.~J. Cand{\`e}s and B.~Recht}, {\em Exact matrix completion via convex
  optimization}, Foundations of Computational mathematics, 9 (2009), p.~717.

\bibitem{clarke1990optimization}
{\sc F.~H. Clarke}, {\em Optimization and nonsmooth analysis}, vol.~5, Siam,
  1990.

\bibitem{combettes2008proximal}
{\sc P.~L. Combettes and J.-C. Pesquet}, {\em A proximal decomposition method
  for solving convex variational inverse problems}, Inverse problems, 24
  (2008), p.~065014.

\bibitem{da2010cone}
{\sc A.~P. Da~Costa and A.~Seeger}, {\em Cone-constrained eigenvalue problems:
  theory and algorithms}, Computational Optimization and Applications, 45
  (2010), pp.~25--57.

\bibitem{de1995exact}
{\sc C.~De~Simone, M.~Diehl, M.~J{\"u}nger, P.~Mutzel, G.~Reinelt, and
  G.~Rinaldi}, {\em Exact ground states of ising spin glasses: New experimental
  results with a branch-and-cut algorithm}, Journal of Statistical Physics, 80
  (1995), pp.~487--496.

\bibitem{deshpande2014cone}
{\sc Y.~Deshpande, A.~Montanari, and E.~Richard}, {\em Cone-constrained
  principal component analysis}, in Advances in Neural Information Processing
  Systems, 2014, pp.~2717--2725.

\bibitem{ding2005equivalence}
{\sc C.~Ding, X.~He, and H.~D. Simon}, {\em On the equivalence of nonnegative
  matrix factorization and spectral clustering}, in Proceedings of the 2005
  SIAM International Conference on Data Mining, SIAM, 2005, pp.~606--610.

\bibitem{douglas1956numerical}
{\sc J.~Douglas and H.~H. Rachford}, {\em On the numerical solution of heat
  conduction problems in two and three space variables}, Transactions of the
  American mathematical Society, 82 (1956), pp.~421--439.

\bibitem{eckstein1992douglas}
{\sc J.~Eckstein and D.~P. Bertsekas}, {\em On the {D}ouglas {R}achford
  splitting method and the proximal point algorithm for maximal monotone
  operators}, Mathematical Programming, 55 (1992), pp.~293--318.

\bibitem{eckstein2015understanding}
{\sc J.~Eckstein and W.~Yao}, {\em Understanding the convergence of the
  alternating direction method of multipliers: Theoretical and computational
  perspectives}, Pac. J. Optim. To appear,  (2015).

\bibitem{fortunato2016community}
{\sc S.~Fortunato and D.~Hric}, {\em Community detection in networks: A user
  guide}, Physics Reports, 659 (2016), pp.~1--44.

\bibitem{friedlander2016low}
{\sc M.~P. Friedlander and I.~Macedo}, {\em Low-rank spectral optimization via
  gauge duality}, SIAM Journal on Scientific Computing, 38 (2016),
  pp.~A1616--A1638.

\bibitem{fujie1997semidefinite}
{\sc T.~Fujie and M.~Kojima}, {\em Semidefinite programming relaxation for
  nonconvex quadratic programs}, Journal of Global Optimization, 10 (1997),
  pp.~367--380.

\bibitem{gabay1975dual}
{\sc D.~Gabay and B.~Mercier}, {\em A dual algorithm for the solution of non
  linear variational problems via finite element approximation}, Institut de
  recherche d'informatique et d'automatique, 1975.

\bibitem{gander1991constrained}
{\sc W.~Gander, G.~H. Golub, and U.~von Matt}, {\em A constrained eigenvalue
  problem}, in Numerical Linear Algebra, Digital Signal Processing and Parallel
  Algorithms, Springer, 1991, pp.~677--686.

\bibitem{gillis2011nonnegative}
{\sc N.~Gillis et~al.}, {\em Nonnegative matrix factorization: Complexity,
  algorithms and applications}, Unpublished doctoral dissertation,
  Universit{\'e} catholique de Louvain. Louvain-La-Neuve: CORE,  (2011).

\bibitem{girvan2002community}
{\sc M.~Girvan and M.~E. Newman}, {\em Community structure in social and
  biological networks}, Proceedings of the national academy of sciences, 99
  (2002), pp.~7821--7826.

\bibitem{glowinski1975approximation}
{\sc R.~Glowinski and A.~Marroco}, {\em Sur l'approximation, par
  {\'e}l{\'e}ments finis d'ordre un, et la r{\'e}solution, par
  p{\'e}nalisation-dualit{\'e} d'une classe de probl{\`e}mes de dirichlet non
  lin{\'e}aires}, Revue fran{\c{c}}aise d'automatique, informatique, recherche
  op{\'e}rationnelle. Analyse num{\'e}rique, 9 (1975), pp.~41--76.

\bibitem{goemans1995improved}
{\sc M.~X. Goemans and D.~P. Williamson}, {\em Improved approximation
  algorithms for maximum cut and satisfiability problems using semidefinite
  programming}, Journal of the ACM (JACM), 42 (1995), pp.~1115--1145.

\bibitem{goldstein2014fast}
{\sc T.~Goldstein, B.~O'Donoghue, S.~Setzer, and R.~Baraniuk}, {\em Fast
  alternating direction optimization methods}, SIAM Journal on Imaging
  Sciences, 7 (2014), pp.~1588--1623.

\bibitem{helmberg2000semidefinite}
{\sc C.~Helmberg}, {\em Semidefinite programming for combinatorial
  optimization}, Konrad-Zuse-Zentrum f{\"u}r Informationstechnik Berlin, 2000.

\bibitem{helmberg1998solving}
{\sc C.~Helmberg and F.~Rendl}, {\em Solving quadratic (0, 1)-problems by
  semidefinite programs and cutting planes}, Mathematical programming, 82
  (1998), pp.~291--315.

\bibitem{helmberg2000spectral}
{\sc C.~Helmberg and F.~Rendl}, {\em A spectral bundle method for semidefinite
  programming}, SIAM Journal on Optimization, 10 (2000), pp.~673--696.

\bibitem{holland1983stochastic}
{\sc P.~W. Holland, K.~B. Laskey, and S.~Leinhardt}, {\em Stochastic
  blockmodels: First steps}, Social networks, 5 (1983), pp.~109--137.

\bibitem{hong2016convergence}
{\sc M.~Hong, Z.-Q. Luo, and M.~Razaviyayn}, {\em Convergence analysis of
  alternating direction method of multipliers for a family of nonconvex
  problems}, SIAM Journal on Optimization, 26 (2016), pp.~337--364.

\bibitem{huang2016consensus}
{\sc K.~Huang and N.~D. Sidiropoulos}, {\em Consensus-admm for general
  quadratically constrained quadratic programming}, IEEE Transactions on Signal
  Processing, 64 (2016), pp.~5297--5310.

\bibitem{jaggi2010simple}
{\sc M.~Jaggi, M.~Sulovsk, et~al.}, {\em A simple algorithm for nuclear norm
  regularized problems}, in Proceedings of the 27th international conference on
  machine learning (ICML-10), 2010, pp.~471--478.

\bibitem{javanmard2016phase}
{\sc A.~Javanmard, A.~Montanari, and F.~Ricci-Tersenghi}, {\em Phase
  transitions in semidefinite relaxations}, Proceedings of the National Academy
  of Sciences, 113 (2016), pp.~E2218--E2223.

\bibitem{jiang2014alternating}
{\sc B.~Jiang, S.~Ma, and S.~Zhang}, {\em Alternating direction method of
  multipliers for real and complex polynomial optimization models},
  Optimization, 63 (2014), pp.~883--898.

\bibitem{judice2007eigenvalue}
{\sc J.~J. J{\'u}dice, H.~D. Sherali, and I.~M. Ribeiro}, {\em The eigenvalue
  complementarity problem}, Computational Optimization and Applications, 37
  (2007), pp.~139--156.

\bibitem{karger1998approximate}
{\sc D.~Karger, R.~Motwani, and M.~Sudan}, {\em Approximate graph coloring by
  semidefinite programming}, Journal of the ACM (JACM), 45 (1998),
  pp.~246--265.

\bibitem{karisch1998semidefinite}
{\sc S.~E. Karisch and F.~Rendl}, {\em Semidefinite programming and graph
  equipartition}, Topics in Semidefinite and Interior-Point Methods, 18 (1998),
  pp.~77--95.

\bibitem{keeling2005implications}
{\sc M.~Keeling}, {\em The implications of network structure for epidemic
  dynamics}, Theoretical population biology, 67 (2005), pp.~1--8.

\bibitem{krislock2012improved}
{\sc N.~Krislock, J.~Malick, and F.~Roupin}, {\em Improved semidefinite
  branch-and-bound algorithm for k-cluster}, Available online as preprint
  hal-00717212,  (2012).

\bibitem{laurent2009sums}
{\sc M.~Laurent}, {\em Sums of squares, moment matrices and optimization over
  polynomials}, in Emerging applications of algebraic geometry, Springer, 2009,
  pp.~157--270.

\bibitem{lee1999learning}
{\sc D.~D. Lee and H.~S. Seung}, {\em Learning the parts of objects by
  non-negative matrix factorization}, Nature, 401 (1999), p.~788.

\bibitem{lee2001algorithms}
{\sc D.~D. Lee and H.~S. Seung}, {\em Algorithms for non-negative matrix
  factorization}, in Advances in neural information processing systems, 2001,
  pp.~556--562.

\bibitem{li2015global}
{\sc G.~Li and T.~K. Pong}, {\em Global convergence of splitting methods for
  nonconvex composite optimization}, SIAM Journal on Optimization, 25 (2015),
  pp.~2434--2460.

\bibitem{lions1979splitting}
{\sc P.-L. Lions and B.~Mercier}, {\em Splitting algorithms for the sum of two
  nonlinear operators}, SIAM Journal on Numerical Analysis, 16 (1979),
  pp.~964--979.

\bibitem{liu2017linearized}
{\sc Q.~Liu, X.~Shen, and Y.~Gu}, {\em Linearized admm for non-convex
  non-smooth optimization with convergence analysis}, arXiv preprint
  arXiv:1705.02502,  (2017).

\bibitem{lu2017nonconvex}
{\sc S.~Lu, M.~Hong, and Z.~Wang}, {\em A nonconvex splitting method for
  symmetric nonnegative matrix factorization: Convergence analysis and
  optimality}, IEEE Transactions on Signal Processing,  (2017).

\bibitem{magnusson2016convergence}
{\sc S.~Magn{\'u}sson, P.~C. Weeraddana, M.~G. Rabbat, and C.~Fischione}, {\em
  On the convergence of alternating direction {Lagrangian} methods for
  nonconvex structured optimization problems}, IEEE Transactions on Control of
  Network Systems, 3 (2016), pp.~296--309.

\bibitem{papadopoulos2012community}
{\sc S.~Papadopoulos, Y.~Kompatsiaris, A.~Vakali, and P.~Spyridonos}, {\em
  Community detection in social media}, Data Mining and Knowledge Discovery, 24
  (2012), pp.~515--554.

\bibitem{pataki1998rank}
{\sc G.~Pataki}, {\em On the rank of extreme matrices in semidefinite programs
  and the multiplicity of optimal eigenvalues}, Mathematics of operations
  research, 23 (1998), pp.~339--358.

\bibitem{poljak1995recipe}
{\sc S.~Poljak, F.~Rendl, and H.~Wolkowicz}, {\em A recipe for semidefinite
  relaxation for (0, 1)-quadratic programming}, Journal of Global Optimization,
  7 (1995), pp.~51--73.

\bibitem{poljak1994expected}
{\sc S.~Poljak and Z.~Tuza}, {\em The expected relative error of the polyhedral
  approximation of the {MAX-CUT} problem}, Operations Research Letters, 16
  (1994), pp.~191--198.

\bibitem{queiroz2004symmetric}
{\sc M.~Queiroz, J.~Judice, and C.~Humes~Jr}, {\em The symmetric eigenvalue
  complementarity problem}, Mathematics of Computation, 73 (2004),
  pp.~1849--1863.

\bibitem{recht2010guaranteed}
{\sc B.~Recht, M.~Fazel, and P.~A. Parrilo}, {\em Guaranteed minimum-rank
  solutions of linear matrix equations via nuclear norm minimization}, SIAM
  review, 52 (2010), pp.~471--501.

\bibitem{rendl2012semidefinite}
{\sc F.~Rendl}, {\em Semidefinite relaxations for partitioning, assignment and
  ordering problems}, 4OR, 10 (2012), pp.~321--346.

\bibitem{rendl2007branch}
{\sc F.~Rendl, G.~Rinaldi, and A.~Wiegele}, {\em A branch and bound algorithm
  for {MAX-CUT} based on combining semidefinite and polyhedral relaxations}, in
  IPCO, vol.~4513, Springer, 2007, pp.~295--309.

\bibitem{rockafellar1974augmented}
{\sc R.~T. Rockafellar}, {\em Augmented lagrange multiplier functions and
  duality in nonconvex programming}, SIAM Journal on Control, 12 (1974),
  pp.~268--285.

\bibitem{shen2014augmented}
{\sc Y.~Shen, Z.~Wen, and Y.~Zhang}, {\em Augmented lagrangian alternating
  direction method for matrix separation based on low-rank factorization},
  Optimization Methods and Software, 29 (2014), pp.~239--263.

\bibitem{shi2000normalized}
{\sc J.~Shi and J.~Malik}, {\em Normalized cuts and image segmentation}, IEEE
  Transactions on pattern analysis and machine intelligence, 22 (2000),
  pp.~888--905.

\bibitem{spingarn1985applications}
{\sc J.~E. Spingarn}, {\em Applications of the method of partial inverses to
  convex programming: decomposition}, Mathematical Programming, 32 (1985),
  pp.~199--223.

\bibitem{sun2015expected}
{\sc R.~Sun, Z.-Q. Luo, and Y.~Ye}, {\em On the expected convergence of
  randomly permuted admm}, arXiv preprint arXiv:1503.06387,  (2015).

\bibitem{udell2016generalized}
{\sc M.~Udell, C.~Horn, R.~Zadeh, S.~Boyd, et~al.}, {\em Generalized low rank
  models}, Foundations and Trends{\textregistered} in Machine Learning, 9
  (2016), pp.~1--118.

\bibitem{wang2015global}
{\sc Y.~Wang, W.~Yin, and J.~Zeng}, {\em Global convergence of admm in
  nonconvex nonsmooth optimization}, arXiv preprint arXiv:1511.06324,  (2015).

\bibitem{wolkowicz2012handbook}
{\sc H.~Wolkowicz, R.~Saigal, and L.~Vandenberghe}, {\em Handbook of
  semidefinite programming: theory, algorithms, and applications}, vol.~27,
  Springer Science \& Business Media, 2012.

\bibitem{xu2012alternating}
{\sc Y.~Xu, W.~Yin, Z.~Wen, and Y.~Zhang}, {\em An alternating direction
  algorithm for matrix completion with nonnegative factors}, Frontiers of
  Mathematics in China, 7 (2012), pp.~365--384.

\bibitem{yinthree}
{\sc W.~Yin}, {\em Three-operator splitting and its optimization applications}.

\bibitem{yuan2005projective}
{\sc Z.~Yuan and E.~Oja}, {\em Projective nonnegative matrix factorization for
  image compression and feature extraction}, in Scandinavian Conference on
  Image Analysis, Springer, 2005, pp.~333--342.

\bibitem{zass2007nonnegative}
{\sc R.~Zass and A.~Shashua}, {\em Nonnegative sparse pca}, in Advances in
  neural information processing systems, 2007, pp.~1561--1568.

\end{thebibliography}
\end{document}